\documentclass[12pt]{article}

\usepackage{color}
\usepackage{amsmath,amsthm,amssymb,url}
\usepackage{graphicx}
\usepackage{forest}
\usepackage{colortbl}

\usepackage{fullpage}
\def\phi{\varphi}

\newcommand{\hm}[1]{{\color{magenta}#1}}
\newcommand{\hy}[1]{{\color{yellow}#1}}
\newcommand{\hr}[1]{{\color{red}#1}}
\newcommand{\hb}[1]{{\color{blue}#1}}
\newcommand{\hc}[1]{{\color{cyan}#1}}
\newcommand{\hg}[1]{{\color{green}#1}}
\newcommand{\ta}{\tilde a}
\newcommand{\tb}{\tilde b}
\newcommand{\tc}{\tilde c}
\newcommand{\td}{\tilde d}
\newcommand{\dek}[1]{{\color{violet} #1}}
\newcommand{\ver}[1]{{\color{blue} #1}}
\renewcommand{\dek}[1]{}
\renewcommand{\ver}[1]{}

\newtheorem{lemma}{Lemma}
\newtheorem{theorem}{Theorem}
\theoremstyle{remark}
\newtheorem{remark}{Remark}
\newtheorem{definition}{Definition}

\newcommand{\EQ}{\mathrm{EQ}}
\newcommand{\DISJ}{\mathrm{DISJ}}
\newcommand{\IP}{\mathrm{IP}}

\renewcommand{\le}{\leqslant}\renewcommand{\ge}{\geqslant}
\newcommand{\rr}{\mathtt{r}}

\title{Half-duplex communication complexity with adversary
  can be less than the classical communication complexity}
\author{
Mikhail Dektiarev\footnote{Moscow State University}
\and
Nikolay Vereshchagin\footnotemark[1] \footnote{National Research University Higher School of Economics, Moscow, Russian Federation}
\footnote{This paper was prepared within the framework of the HSE University Basic Research Program. 
}
}
\date{}
\begin{document}
\maketitle
\begin{abstract}
Half-duplex communication complexity with adversary 
was defined in [Hoover, K., Impagliazzo, R., Mihajlin, I., Smal, A. V. Half-Duplex
Communication Complexity, ISAAC 2018.]
Half-duplex communication protocols generalize classical protocols
defined by Andrew Yao in [Yao, A. C.-C. Some Complexity Questions Related to Distributive Computing
  (Preliminary Report),  STOC 1979]. It has been unknown so far
whether the communication complexities defined by these models
are different or not. In the present paper we answer
this question: we exhibit a function  whose  
half-duplex communication complexity with adversary
is strictly less than its classical communication complexity.
\end{abstract}

\section{Introduction}
In the classical model of communication complexity introduced
by Andrew Yao in~\cite{yao}, we consider a cooperative game between
two players, Alice and Bob, who want to compute
$f(x,y)$ for a given function $f$.
Alice knows only $x$ and  Bob knows only $y$.
To this end, Alice
and Bob can communicate sending to each other
messages, one bit per round.
An important property of this
communication model is the following:
in each round one player sends a bit message while the other player
receives it. Algorithms for Alice and Bob computing $f$
are called \emph{communication protocols}. The \emph{depth} of a protocol is defined as 
the number of bits it communicates in the worst case.
The minimum possible depth of a protocol to compute $f$
is called the \emph{communication complexity} of $f$.

This model was generalized in~\cite{hims} to a model describing
communication over the so called half-duplex channel.
A well-known example of half-duplex
communication is talking via walkie-talkie:
one has to hold a ``push-to-talk'' button to
speak to another person, and one has to release it to listen.
We consider a communication model where
players are allowed to speak simultaneously.
If two persons try to speak
simultaneously then they do not hear each other  and 
both messages are lost.

The communication protocol over a half-duplex channel is also
divided into rounds. To make such a division, we assume that
the players have synchronized clocks.
Every round each player chooses one of three
actions: send 0, send 1, or receive.
Thus we distinguish three types of rounds.
\begin{itemize}
\item If one players sends a bit
  and the other one receives,
  then we call this round \emph{normal}
  or \emph{classical}, in such a round the receiving player
  receives the sent bit.
\item If both players send bits, then the round is called \emph{spent}.
  Both bits are lost (such a situation occurs
  if both parties push ``push-to-talk'' buttons
  simultaneously).
\item If both players receive, then the round is called \emph{silent}.
  In such rounds players receive arbitrary
  bits that may be different. (We can think that those bits are chosen by a malicious adversary.)
\end{itemize}	
More specifically, the described model is called 
the communication model with  \emph{adversary}.
In~\cite{hims} there were considered also two other models: the half-duplex model
\emph{with silence}
(in a silent round both players receive a special symbol \texttt{silence}
and hence know that a silent round occurred) and  the half-duplex model
\emph{with zero} (in a silent round both players receive 0).

More specifically, the \emph{half-duplex complexity with adversary} of a function is defined as minimum
number of rounds needed to compute the function on its domain, assuming that the  adversary can choose
any bits in silent rounds. On can show that the half-duplex
communication complexity with adversary is sandwiched between the classical complexity and a half of it~\cite{hims}.

The original motivation to study these kinds of communication models arose from the
question of the complexity of Karchmer-Wigderson games~\cite{kw}
for multiplexors.
A detailed exposition of this motivation and recent results in this direction can be
found in~\cite{meir2020,meir2023,ignat,wu}. Here we will just
present a toy example in which
half-duplex complexity  arises quite naturally.
Assume that
for each parameter $a$ (ranging over a finite set)
a function $f_a:X\times Y\to \{0,1\}$ is given. 
Our goal is to prove that for some $a$
the classical communication complexity of $f_a$ 
is larger that a certain number $d$.
Assume also that we can prove the same lower bound
$d$ for the classical communication complexity of the \emph{multiplexor game}
defined as follows: Alice gets a pair $(x,a)$, Bob gets a pair
$(y,b)$ and they want to compute 
$f_a(x,y)$, if $a=b$. Otherwise, if $a\ne b$,
then their protocol may return any result.
It may seem that from this we can deduce
the sought lower bound for $f_a$ for some $a$. Indeed, by way of contradiction,
assume that for every $a$
there is a classical depth-$d$ communication protocol $\Pi_a$ that computes $f_a$.
Then consider the following depth-$d$ communication protocol for the multiplexor game: 
if  Alice gets a pair $(x,a)$ and Bob a pair $(y,b)$, then
Alice finds the lex first depth-$d$ protocol  for $f_a$
and  Bob the  lex first depth-$d$ protocol for
$f_b$. Then they run the found protocols.
If $a=b$, then they run the same protocol,
which outputs $f_a(x,y)$. Otherwise, if $a\ne b$, their protocols
can output different results. This problem is easy to overcome:
Alice sends her output to Bob, which costs only one extra bit of
communication. However, there is more complicated problem:  protocols  $\Pi_a$ and  $\Pi_b$ may simultaneously receive or send bits.
Thus actually the constructed protocol is a half-duplex and not classical communication protocol for multiplexor game.
Thus we need a  lower bound $d$ for half-duplex communication complexity of
the multiplexor game!

There are several examples of functions whose  
half-duplex complexity with silence and with zero is less than the
classical communication complexity. This happens for 
the following three functions, which are widely studied in Communication complexity:
Equality function $\EQ_n$, Disjointness function $\DISJ_n$ and Inner product $\IP_n$. These functions 
are defined on bit strings of length $n$ as follows:
$$
\EQ_n(x,y)=\begin{cases}
1& \text{if }x=y,\\
0& \text{otherwise},
\end{cases}\quad
\DISJ_n(x,y)=\bigwedge_{i=1}^n \lnot(x_i\land y_i),
\quad
\IP_n(x,y)=\sum_{i=1}^n( x_i\land y_i)\pmod2.
$$
The classical communication complexity of all these functions is $n+1$. 
The best known upper bounds for their   half-duplex complexity are shown in the following table:
\begin{center}
\begin{tabular}{|l|c|c|c|}
\hline
& $\EQ_n$&  $\IP_n$&  $\DISJ_n$\\ \hline
Half-duplex complexity with silence & $n\log_52+o(n)$& $n/2+2$ \cite{glad}& $n/2+2$\\ \hline
Half-duplex complexity with 0 &$n\log_32+o(n)$& $7n/8+O(1)$ \cite{kogan}& $3n/4+o(n)$ \cite{dem2021}\\
\hline
Half-duplex complexity with adversary &$n+1$& $n+1$& $n+1$\\
\hline
\end{tabular}
\end{center}
The best known lower bounds are the following:
\begin{center}
\begin{tabular}{|l|c|c|c|}
\hline
& $\EQ_n$&  $\IP_n$&  $\DISJ_n$\\ \hline
Half-duplex complexity with silence & $n\log_52$& $n/2$& $n\log_52$ \cite{dem2021}\\\hline
Half-duplex complexity with 0 &$n\log_32$& $n\log_{\frac2{3-\sqrt5}}2$& $n\log_32$ \cite{dem2021}\\
\hline
Half-duplex complexity with adversary &$n\log_{2.5}2$& $n\log_{7/3}2$& $n\log_{2.5}2$ \cite{dem2021}\\
\hline
\end{tabular}
\end{center}
Some cells of these tables contain references showing who established the corresponding bounds.  The bounds with no reference are due to~\cite{hims}.

However no examples of functions for which 
half-duplex complexity with adversary is less than the classical complexity
were known so far. Moreover, in~\cite[page 67]{smal} it was conjectured that
these two complexities coincide.
In this paper we exhibit a function $f$ with a constant gap between
these complexities (the half-duplex complexity with adversary is 5, and the classical complexity is 6). Also
we exhibit a family of partial function $g_n$ with a logarithmic gap between these
complexities. More specifically,
the function $g_n$ is defined on
length-$n$ strings over a fixed alphabet, its half-duplex complexity with adversary
is at most $n$ while the classical complexity is at least $n+\log_2n$. 

To prove our lower bounds,  that is, the bound 6 for the classical complexity of $f$
and the bound $n+\log_2n$ for the  
classical complexity of $g_n$,
we use some novel techniques.
The first bound cannot be proven via the common techniques, which
is based on partitions of the matrix of $f$ into monochromatic rectangles.
A matrix of a function $f:X\times Y\to Z$ is a matrix 
with $|X|$ rows numbered by elements of $X$ and $|Y|$ columns numbered by elements of $Y$.
Its element in $x$th row and $y$th column  is equal to $f(x,y)$. A  monochromatic rectangle in that matrix
is any  subset of $X\times Y$ that has the form $A\times B$ and on which $f$ is constant.
One can show that any communication protocol of depth $d$ computing $f$ provides
a partition of that matrix into at most $2^d$ monochromatic rectangles.
Hence the classical complexity of $f$ can be lower bounded by 
the binary logarithm of minimal possible number of rectangles in a partition of
matrix  of $f$ into monochromatic rectangles.
Usually, the lower bounds for the size of such partitions
are obtained by ``fooling sets''.
A set $F\subset X\times Y$ is called \emph{fooling} for $f$ if every monochromatic rectangle can
cover at most  one pair in $F$. Obviously, in this case the matrix of $f$
cannot be partitioned in less than $|F|$ monochromatic rectangles. So to prove that
classical complexity of $f$ is larger than $d$, it is enough to find a fooling set for $f$ of size larger than $2^d$.

However, the matrix of our function $f$, whose classical complexity is at least 6, 
\emph{can} be partitioned into $2^5$ monochromatic rectangles, thus this method 
fails to show that the classical complexity of $f$ is at least 6. To prove the desired lower bound,
we will use the following feature of partitions that are derived from communication protocols:
every such partition can by obtained by a sequence of horizontal and vertical divisions. More precisely,
we start with trivial partition $P=\{X\times Y\}$ and on each step we partition any rectangle from $P$  
either horizontally (a rectangle $A\times B$ is replaced by two rectangles $A_0\times B$
and $A_1\times B$), or vertically (a rectangle $A\times B$ is partitioned into $A\times B_0$
and $A\times B_1$).  If initial protocol has depth $d$ then
any rectangle in the resulting partition is obtained from the initial rectangle by at most $d$ such divisions.

We show that for every horizontal partition
of the matrix of $f$ into two sub-matrices 
$V,W$ 
either $U$ or $V$ has a fooling set of size $17>2^4=16$. 
This is done using the following method:
in the matrix of $f$ we distinguish the so called ``fooling rectangles''.
They are pairwise disjoint and have the following feature:
if we pick from every fooling rectangle any cell, then
the resulting set of cells is a fooling set in the matrix.
The number of fooling rectangles is $25$,
thus this implies only that the matrix has a fooling set of size  $25\le 2^5$.
We show however that  for every horizontal partition
of the matrix into two sub-matrices 
$V,W$ one of $V,W$ intersects at least  $17$ fooling rectangles.
To show that, we consider the graph
whose vertices are fooling rectangles
and edges connect those rectangles $R_1,R_2$ that share a row.
Then we prove that that graph has certain  expanding property.
More specifically, every subset of  $25-16=9$ vertices in the graph has at least $17$ neighbors. This implies the desired property. Indeed,
let $\mathcal R$ stand for the family of all fooling rectangles that
share a row with 
$V$. If $|\mathcal R|\le 16$, then its complement $\overline{\mathcal R}$ has at least $9$ rectangles
and hence at least $17$ neighbors and all them share a row with $W$! 

Then we prove similar statement for vertical partitions, using other fooling rectangles.

To prove the lower bound  $n+\log_2n$
for 
communication complexity of $g_n$,
we generalize the notion of a monochromatic rectangle
to partial functions. This time we call
a rectangle $A\times B$  monochromatic, if
$f$ is constant in each row and in each column of
  $A\times B$ (in the points where it is defined).
Note that  if $f$ is not total then it may still be non-constant in the entire
rectangle  $A\times B$. However, for total functions this definition coincides with
the classical definition of monochromatic rectangles.
The matrix of our function $g_n$ can be \emph{covered} by $2^n$ monochromatic rectangles.
However, we show that every its \emph{partition} into monochromatic rectangles
has at least $n2^n$ rectangles.

The next section contains main definitions. In Section~\ref{s3} 
we study partial functions, and in Sections~\ref{s4} and~\ref{s5} we study total functions. 
Section~\ref{s4} is a ``warming up'' section, we exhibit there an example
of a total function with a gap between  classical complexity and half-duplex complexity with the so called
``honest'' adversary (an honest adversary sends the same bits to both players in every silent round, those bits may
depend on the round). In Section~\ref{s5} we present our function $f$ with a gap
between classical and  half-duplex complexities (6 vs.  5).

\section{Preliminaries}\label{s2}
\begin{definition}[\cite{yao}]
A (classical) communication protocol
to compute a partial function $f:X\times Y\to Z$
is a finite rooted tree $T$ each of whose 
internal nodes (including the root) has two children.
Its leaves are labeled by elements of $Z$,
and each internal node $v$ (including the root) is labeled either by letter \texttt{A} and 
by some function from $X$ to $\{0,1\}$,
or by letter \texttt{B} 
and by some function from $Y$ to $\{0,1\}$.
Besides, one outgoing edge from $v$ is labeled by 0 and the other one is labeled by 1.
\end{definition}
A protocol of depth 3 to compute a function $f:\{0,1\}\times \{0,1\}\to\{1,2,3,4\}$
is shown on Fig.~\ref{pi-4}.
\begin{definition}
The computation of a protocol for the input pair $(x,y)\in X\times Y$ 
runs as follows. At the start, the players, Alice who has $x$ and Bob who has $y$, place their tokens on the root
of the tree. Then they move them along edges in the following way.
If the tokens are in an internal vertex 
 $u$ labeled by $\mathtt{A}$ and a function $h:X\to\{0,1\}$, then they move the tokens
along the edge labeled by $h(x)$. To this end Alice sends $h(x)$ to Bob letting him know where to move his token.
If the tokens are in an internal vertex 
 $u$ labeled by $\mathtt{B}$ and a function $h:Y\to\{0,1\}$, then the tokens move
along the edge labeled by $h(y)$ (Bob sends $h(y)$ to Alice).
Finally, if the tokens come to a leaf, then the label of that leaf is
the result of the computation.  The rules guarantee that their tokens are always in the same node,
thus Alice and Bob output the same results.
\end{definition}

\begin{definition}
The sequence of the sent bits (= the leaf to which the token comes)
is called the \emph{transcript} of the communication.
A communication protocol computes a partial function
$f:X\times Y\to Z$, if for all input pairs $(x,y)$ from the domain of $f$
the result of the computation is equal to $f(x,y)$.
The minimal depth of a communication protocol to compute  $f$
is called \emph{the communication complexity} of $f$.
\end{definition}

\begin{figure}[ht]
\begin{center}
\includegraphics[scale=1]{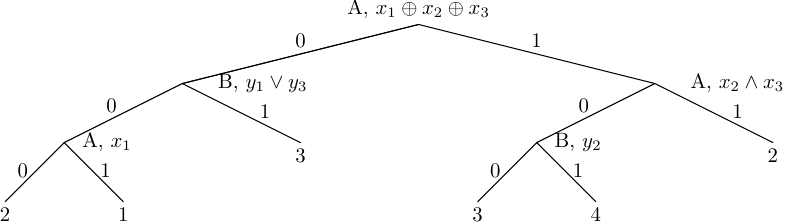}
\end{center}
\caption{A protocol of depth 3 to compute a function $f:\{0,1\}^3\times \{0,1\}^3\to\{1,2,3,4\}$.
Each internal node is labeled by a letter indicating the turn to move and by a function
computing the bit to send. For instance,
if Alice has the string $x=010$ and Bob the string $y=110$, then
in the first round Alice sends 1, in the second round Alice sends 0 and in the third round Bob sends 1.
Then both parties output 4.
}\label{pi-4}
\end{figure}

In the sequel we use the following well known facts  
(see for example~\cite{kn}) about communication
protocols to compute total functions:
\begin{itemize}
\item Let  $\Pi$ be a communication protocol. 
  For every node  $u$ of $\Pi$ the set of the input pairs $(x,y)$ such that the token comes through $u$
  is a (combinatorial) rectangle, that is, it has the form $A_u\times B_u$ for some $A_u\subset X$ and  $B_u\subset Y$.
  If $u_0$ and $u_1$ are children of $u$ and $u$ is labeled by A then 
  $A_u=A_{u_0}\cup A_{u_1}$ and   
  $B_u=B_{u_0}=B_{u_1}$. In this case we say that rectangles
  $A_{u_0}\times B_{u_0}$ and $A_{u_1}\times B_{u_1}$
  are obtained from $A_u\times B_u$ by a horizontal division.
  If $u$ is labeled by B then 
  $B_u=B_{u_0}\cup B_{u_1}$ and   
  $A_u=A_{u_0}=A_{u_1}$. In this case we say that $A_{u_0}\times B_{u_0}$ and $A_{u_1}\times B_{u_1}$
  are obtained from $A_u\times B_u$ by a vertical division.

\item If $f$ is constant in a rectangle $R$, then $R$ is called a 
\emph{monochromatic rectangle for $f$}.
\item
If a protocol $\Pi$ computes a total function $f$,
then the rectangle corresponding to any leaf of the  protocol,
is monochromatic for $f$. Hence
to each protocol computing 
$f$ we can assign a partition of $X\times Y$ into monochromatic rectangles.
That partition has at most  $2^{\text{depth of the protocol}}$ rectangles.
Thus  
communication  complexity is at least the binary logarithm
of the minimal size of the partition of $X\times Y$ into monochromatic rectangles.
\item 
A set $F\subset X\times Y$ is called a 
\emph{fooling set} for  function $f:X\times Y\to Z$,
if for all different pairs $(x,y),(u,v)\in F$ 
not all values $f(x,y),f(u,v),f(x,v),f(u,y)$ 
are equal. In this case every partition of  $X\times Y$
into monochromatic rectangles has at least $|F|$ rectangles.
\end{itemize}

The definition of a half-duplex communication protocol is more complicated.

\begin{definition}[\cite{hims}]\label{def-hd}
A half-duplex communication protocol 
to compute a partial  function $f:X\times Y\to Z$
is a pair of rooted trees $T_A,T_B$. Each internal node of
both trees has 4 children.
Each internal vertex $v$ of  $T_A$ is labeled by a function from 
$X$ into the 3-element set consisting of \emph{actions}, 
$\text{``send 0''}, \text{``send 1''}, \text{``receive''}$.
Edges outgoing from $v$ are labeled by  
\emph{events}  
\text{``sent 0''}, \text{``sent 1''}, \text{``received 0''}, \text{``received 1''}.
Similarly, each internal vertex of  $T_B$ is labeled by a   
function from  $Y$ into the set of actions, 
and outgoing edges are labeled by events. The leaves of both trees are labeled by elements of $Z$.

The computation for an input pair $(x,y)\in X\times Y$ 
runs as follows.
First we calculate for each internal vertex of both trees
the action of the player applying the corresponding function to her/his input
($x$ or $y$). Second, 
each player puts a token on the root of her/his tree and
moves it according to the following rules. If the tokens
are in the vertices $u,v$, and the actions of the players in those nodes
are $a,b$, respectively, then the tokens are moved as follows:

\medskip
\begin{tabular}{|c|c|c|c|}
\hline
& & Alice's token moves&  Bob's token moves\\
$a$ & $b$& along the edge labeled by &  along the edge labeled by\\
 &  & the event& the event\\
\hline
send $i$ & receive & ``sent $i$'' &``received $i$'' \\
\hline
receive& send $j$ & ``received  $j$'' &``sent    $j$'' \\
\hline
send    $i$& send $j$ & ``sent  $i$'' &``sent  $j$'' \\
\hline
receive   & receive &``received $j$'' &``received $i$'' \\
\hline
\end{tabular}

\medskip
If both players chose to receive (the last row of the table) then the round is called
``silent''.
In this case the received values  $i,j$ are chosen by an adversary who decides
where the tokens are moved to.
 
The computation stops when one of the token reaches a leaf. 
The result of the computation is defined as follows. If the other token is
not in a leaf then the result is undefined.
If both tokens reach leaves but their labels are different,
then the result is undefined as well.
Finally, if those labels coincide, then that label is the result of
the computation.
\end{definition}

An example of a 1-round half-duplex protocol is shown on Fig.~\ref{pi-5}.
\begin{figure}[ht]
\begin{center}
\includegraphics[scale=1]{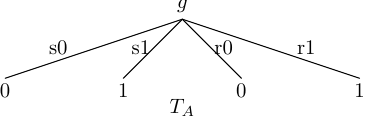}
\hskip 1cm
\includegraphics[scale=1]{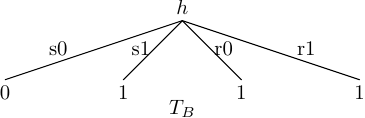}
\end{center}
\caption{A 1-round half-duplex protocol.
Here we assume that $X=Y=\{0,1,2\}$,  $Z=\{0,1\}$.
The event ``sent $i$'' is abbreviated as s$i$
and similarly  the event ``receive $i$'' is abbreviated as r$i$.
The letter $g$ denotes the mapping $0\mapsto \text{send 0}, 1\mapsto \text{send 1}, 2\mapsto \text{receive}$,
and  $h$ denotes the same mapping.
This protocol computes the function  $f:\{0,1,2\}\times \{0,1,2\}\to\{0,1\}$
which is defined only on the pairs $(0,0), (1,1), (1,2)$.
On the first and the second pair both players send a bit which is lost, however they output the same result.
On the third pair Alice sends 1, which is received by Bob, who actually does not care,
since he outputs 1 any way.}\label{pi-5}
\end{figure}

\begin{definition}
A half-duplex protocol computes a function $f$,
if for all pairs $(x,y)$ in its domain and for all choices of 
the received bits in silent rounds (in different silent rounds different pairs $(i,j)$ can be chosen)
the computation ends with the result $f(x,y)$.  If $f(x,y)$ is undefined, then
the computation can end with any result or without any result.
\end{definition}
\begin{definition}
  There is a natural variation of computation via a  half-duplex protocol, in which the adversary
  sends the same bits to Bob and Alice in each silent round (that bit may depend on the round).
  Such an adversary is called \emph{honest}.
  To distinguish honest adversaries from general type adversaries we will call the latter \emph{malicious}.
\end{definition}

To distinguish normal communication protocols and complexity
from half-duplex ones, we will call the former \emph{classical}. 

\begin{remark}
Classical protocols can be viewed as a particular case of half-duplex protocols. To convert a classical protocol
into a half-duplex one, we first transform its tree $T$ as follows.
For each internal $v$ vertex of $T$ we make two copies of the tree with the root in the left child of $v$ and
two copies of the tree with the root in the right child of $v$.
The edge going to the first copy of the left child is labeled by the event
``sent  0'', and the edge going to the second copy of the left child is labeled by the event
``received  0'', and similarly for the right child.
This transformation should by applied to all internal nodes in any order. 
The trees $T_A,T_B$ are equal to the resulting trees.
If an internal node $v$ is labeled by $\mathtt{A},h$ in $T$,
then in the Alice's tree  $T_A$ all copies of $v$ are marked
by the function that maps  $x$ to the action ``send $h(x)$'', and in Bob's tree
 $T_B$ --- by the constant function
``receive''. Similarly,  if $v$ is labeled by $\mathtt{B},h$ in $T$,
then in  Bob's tree all copies of  $v$ are labeled by the function mapping
$y$ to ``send $h(y)$'', and in  Alice's tree --- by the constant function 
``receive''. By construction, in the computation of the resulting 
protocol, there are no silent and spent rounds.
\end{remark}

Obviously, half-duplex complexity with honest adversary is less than or equal to the
half-duplex complexity with malicious adversary, and the latter is less than or equal to the
classical communication  complexity. One can also show that half-duplex complexity
with both adversaries is larger than or equal to  the half of the classical complexity~\cite{hims}\footnote{Actually,
in~\cite{hims}   this was shown for half-duplex complexity with zero in place of the half-duplex complexity with
  honest adversary. Obviously, the former one is at most the latter one.}.

\section{A separation of half-duplex and classical communication complexities for partial functions
}\label{s3}

\subsection{An example}\label{s31}
\begin{lemma}
There is a partial function $g$ whose half-duplex complexity is $1$ and classical  complexity is at least $2$.
\end{lemma}
\begin{proof}
Let $X = Y = \{0, 1, \rr\}$, $Z =\{0\rr,1\rr,\rr0,\rr1\}$ and the 
partial function $g: X\times Y \to Z$ is defined by the following table:
$$
\begin{array}{c|c|c|c|}
    & 0 & 1 & \rr \\ \hline
  0 &   &  & 0\rr \\ \hline
  1 &   &  & 1\rr \\ \hline
  \rr & \rr0 & \rr1 &  \\ \hline
\end{array}
$$
Here rows are labeled be Alice's inputs $x$ and columns by Bob's inputs $y$. The 
entry of the table in row $x$ and column $y$ is blank, if 
$g(x,y)$ is undefined, and is $g(x,y)$ otherwise. 

The classical complexity of $g$ is at least $2$: in any protocol of depth $1$ either
Alice sends a bit (and hence Alice outputs the same result for input pairs $(\rr,0)$ and $(\rr,1)$), or Bob
sends a bit, and hence Bob outputs the same results for the input pairs $(0, \rr)$ and $(1,\rr)$.

The half-duplex complexity of $g$ is at most $1$. The protocol is the following: 
Alice on the input $0$ sends $0$ and outputs $0\rr$, and on the input $1$, she sends $1$
and outputs $1\rr$. Finally, on the input $\rr$ she receives and after receiving $z$ she outputs $\rr z$.
Bob acts in a similar way: on the input $0$ he sends $0$ and outputs $\rr0$,
on the input $1$ he sends  $1$ and outputs $\rr1$, and on the input $\rr$ he receives and after receiving  $z$ he outputs $z \rr$.

Let us show that this protocol indeed computes $g$. Assume that Alice has $0$ and Bob has $\rr$.
Then Alice sends 0 to Bob, Bob receives 0 and both output $0\rr$.  
A similar analysis works for three other input pairs where the function is defined, $(1,\rr)$,  $(\rr,0)$, and $(\rr,1)$.

For the pairs where $g$ is undefined (e.g. (0, 1)) Alice and Bob may output different results, but we do not care.  
\end{proof}

Consider the following generalization of this  function, which provides a linear  gap
between half-duplex and classical complexities.
Let $X = Y = \{0, 1, \rr\}^n$, $Z =\{0\rr,1\rr,\rr0,\rr1\}^n$ and define $g$ as follows:
$g(x,y)$ is defined if $x_i = \rr \wedge y_i \neq \rr$ or $x_i \neq \rr \wedge y_i = \rr$ for all $i$. In this case the 
$i$-th symbol of $g(x,y)$ equals $\rr y_i$ in the first case and equals
$x_i\rr $ in the second case.
That is, if  Alice's and 
Bobs inputs are ``consistent''
(in each position one party has  $\rr $ and the other has 0 or 1), then the
function is defined and otherwise it is not.

\begin{theorem}
  (a) The 
  half-duplex complexity of 
  $g$ is at most  $n$.
  (b) The 
  classical complexity of 
   $g$ is at least $2n$
\end{theorem}
\begin{proof}
  (a) In $i$th round, independently on what has happened in previous rounds,  
  each party receives if her/his $i$th symbol is $\rr $, otherwise she/he sends
   $i$th symbol of the input.

  As $i$th symbol of the output 
 Alice takes  $x_i\rr $, if in $i$th round she sent a bit  (equal to  $x_i$), 
 and
 $\rr y_i$, if she received  $y_i$.
Bob acts in a similar way.

Let us show that this protocol indeed computes $g$. Assume that Alice gets $x$,
Bob gets $y$ where $x$ and $y$ are consistent. 
Then in $i$th round  either Alice sends $x_i$ and Bob receives $x_i$ and then both 
take $x_i\rr $ as $i$th symbol of the output, or 
Bob sends $y_i$ and Alice receives $y_i$ and then both 
take $\rr y_i$ as $i$th symbol of the output. 
Thus  the $i$th symbol of the output of both parties
is equal to  the $i$th symbol of $g(x,y)$ for all $i$.

(b) Note that the 
function $g$ has $4^n$ possible values hence
any classical  communication
protocol computing $g$ must have at least $4^n$ leaves. Thus its 
classical  communication complexity is at least
$\log_2 4^n=2n$.
\end{proof}

\subsection{Local communication complexity}

There is a way to compute the function $g$ from the previous section 
by communicating at most   $n + \lceil \log_2(n + 1)\rceil$ bits (where $n$ is the length $n$ of the inputs):
\begin{enumerate}
\item  Alice computes the number $k$ of symbols $\rr$ in her input.
\item She then sends $k$ to Bob in    $\lceil\log_2 (n + 1)\rceil$ bits.
Note that $k\in\{0,1,\dots,n\}$, thus $k$   indeed can be encoded
in this many bits.
\item Bob, after receiving the number $k$, halts with any result if the number of letters  $\rr $ in his input is different 
from  $n-k$. 
\item Then Alice sends $(n - k)$ bits equal to symbols $0$ and $1$ in her input 
keeping their order.
\item  Unless he has halted, Bob receives  $n - k$ bits sent by Alice and then sends  $k$ bits equal to 
symbols $0$ and $1$ in his input, keeping their order.
\item Alice, after receiving a string  $z$, outputs the string $a$ of length $n$ defined as follows. If $x_i\ne \rr $,
then $a_i=x_i\rr $. The remaining symbols of $a$ have the form
$\rr z_j$ where the bits of $z$ are inserted into $a$ in the same order in which they are arranged 
in  $z$.
\item Bob computes his output string in a similar way.
\end{enumerate}

Why this protocol does not contradict to the lower bound $2n$ 
for the communication complexity of the
function $g$? This is because this protocol is not a communication protocol according 
to Definition~\ref{def-hd}. Indeed, Alice and Bob form their outputs based not only on the transcript
of the computation but also on their inputs. As we will see later, for total functions,
this possibility does not yield anything new. However, for partial 
functions this possibility can decrease communication
complexity, as our example shows. A simpler example of this phenomenon
will be given later. By this reason for partial functions a more adequate communication model should allow
players to use their inputs when computing the output.

As far as we know, such a communication model appeared quite recently in~\cite[page 58]{nolin}.
Let us define it more formally.

\begin{definition}[\cite{nolin}]
A local communication protocol for a  function $f:X \times Y \to Z$ is a  
communication protocol in which each leaf is labeled by a pair of functions $X \to Z$ and  $Y \to Z$.
If a computation on an input pair $(x, y)$ has ended in a leaf labeled by a pair of functions $(p,q)$, 
then the result of computation is defined as $p(x)$ provided $p(x) = q(y)$, and is undefined otherwise.
A local communication protocol computes a partial function $f:X \times Y \to Z$ if for all input pairs  $(x,y)$ 
from its domain the  protocol outputs 
$f(x,y)$ (the result can be both defined and undefined
for the pairs outside the domain of the function).
Local 
communication complexity of a
function is the minimal depth of a local communication protocol computing
the function. 
\end{definition}

Local 
half-duplex communication protocols 
are defined in a similar way. When we want to distinguish local 
protocols 
from normal ones, we call the latter ones
\emph{global}.

\subsection{An example}
Consider the 
partial function $f: \{0, 1\}^n \times \{0, 1\}^n \to \{0, 1\}^n$
defined as follows:    
$$
f(x,y)=\begin{cases}
x, &\text{if }x=y\\
\text{undefined}& \text{if }x\ne y
\end{cases}
$$
Since 
$f$ takes $2^n$ different values, any global (classical) protocol
for it has at least $2^n$ leaves and hence its depth is at least $n$.
Thus the global (classical) complexity of 
$f$ is at least 
$n$. On the other hand the local
(classical) complexity of 
$f$ is  $0$: label the root of the tree
with pair $(h, h)$, where $h$ is the identity 
function $h(x) = x$ (each player just outputs her/his input).

It is important in
this example that the
function is not total.
For total  functions there is no such example
both in  
classical and half-duplex models.

\subsection{Local and global complexity of total functions}

\begin{lemma} For classical
  communication    complexity:
  the local complexity of any total 
  function is equal to its global
  complexity.
\end{lemma}

\begin{proof}
Assume that a local  (classical)
communication protocol for $f$ is given and let
$l$ be any its leaf. Let $R_l=A_l\times B_l$
denote the rectangle consisting of all
input pairs $(x,y)$ for which the computation
reaches  $l$. 
We claim that $f$ is constant on $R_l$.

Indeed, let $p_l,q_l$ denote the 
functions which label the leaf $l$.
Then for all  $(x,y)\in R_l$ and $(x',y')\in R_l$ we have  
$$
f(x,y)=q_l(y)=p_l(x')=f(x',y').
$$
Here the first equality holds, since Bob outputs the correct result on the pair $x,y$,
the second equality holds, since Alice and Bob output same results on the pair $x',y$,
which also belongs to $R_l$, and 
the third equality holds, since Alice  outputs the correct result on the pair $x',y'$,
which also belongs to $R_l$.

Hence any local   protocol
to compute any total 
function can be converted into a
classical protocol of the same depth
just by changing the label 
$(p_l,q_l)$ to the constant value of $f$ on $R_l$.
\end{proof}

A similar fact holds for
half-duplex complexity.

\begin{lemma}
  For half-duplex communication    complexity: the local complexity of any  
total function is equal to its global complexity.
\end{lemma}
\begin{proof}
The proof is similar to the proof of the previous lemma. However
for half-duplex protocols the connection between rectangles and leaves
in the protocol trees is more complicated. Therefore  we will not establish this connection
explicitly. Instead we will use a direct argument.

Again we show that any local protocol 
to compute a total function is essentially global.
More precisely, all computations ending in a leaf output the same result. 
Thus we can replace the 
function labeling each leaf in $T_A$ by the constant it outputs
and then we can make the same thing for $T_B$. In this way we transform
a local protocol to the global one.

Consider a local communication protocol $(T_A, T_B)$ to compute a total function $f$.
Assume that for an input pair $(x, y)$ the computation can end in the 
leaves $a, b$ (we use the word ``can'', since now the leaves depend on the choices of
the adversary). Assume also that for another input pair
$(x', y')$ the computation of 
Alice can end in the same leaf $a$.
Let us show that for the input pair
$(x', y)$ the computation can also end in leaves
$(a, b)$.

Let $a_i$ and
$b_i$ denote $i$th vertices on the paths from the roots to  $a$ and $b$, respectively.
Since Alice's computation on both inputs
$(x, y)$ and $(x', y')$ can end in the leaf $a$, for all $i$  Alice's action in the vertex 
$a_i$ for inputs 
$x$ and 
$x'$ is the same.
Thus, if the adversary's action
in the computation on the input pair
$(x', y)$ is identical to that on the input $(x, y)$, when the computation ends in $(a, b)$,
then both players in each round perform the same actions on inputs $(x, y)$ and
$(x', y)$ and all events are identical as well. Hence the computation on the input
$(x', y)$ can also end in the same leaves $a, b$.

Let the leaf
$a$ is labeled by the 
function $p$ and the leaf
$b$ by the function $q$.
Then  $p(x) = f(x, y) = q(y) = f(x', y) = p(x')$. (It is important here that $f$ is a total 
function, as otherwise 
$ f(x', y)$ can be undefined and hence
$q(y)$ may differ from $p(x')$.)

Thus each leaf of  $T_A$ is labeled by a function that is constant on all $x$'s for which the computation can end in that leaf for some $y$.
\end{proof}

\subsection{Local classical complexity of the 
function  $g$}

For partial functions, local protocols seem to be a more natural communication model than the global ones. 
In the Introduction, before we have given formal definition, we talked about 
the lower bound $n+\log_2n$ for classical communication complexity of the function $g_n$.
We meant there the local complexity. 
Now we will prove this statement.

\begin{theorem}
(a) The local classical complexity of $g$ is at most $n + \lceil \log_2(n + 1)\rceil$.

(b) The local classical complexity of $g$ is at least $n + \lceil\log_2 n\rceil - \log_2 3 + o(1)$.
\end{theorem}
\begin{proof}
(a) This statement was proved above.
(b) This statement is more difficult to prove.
\begin{lemma}
If the local classical complexity of a partial function $g$ is at most $c$, than its matrix can be partitioned into at most $2^c$ rectangles with the following property: if inputs $(x, y)$ and $(x, y')$ are in the same rectangle, and $g$ is defined on both of them, then $g(x, y) = g(x, y')$, and similarly, if inputs $(x, y)$ and $(x', y)$ are in the same rectangle, and $g$ is defined on both of them, then $g(x, y) = g(x', y)$.
(Rectangles with this property will be called \emph{monochromatic}.)
\end{lemma}
\begin{proof}
Leaves of classical communication protocol define matrix partition into rectangles.
Each of these rectangles has the mentioned property: if a function $p$ is written in Alice's leaf, and $g$ is defined on $(x, y)$ and $(x, y')$, then we necessarily have $g(x, y) = p(x) = g(x, y')$. Similarly for Bob.

If the depth of the protocol is $c$, then the number of leaves is at most $2^c$, and its rectangles form the required partition.
\end{proof}

We will call any input of the form $(\{0, 1\}^k \rr ^{n - k}, \rr ^m \{0, 1\}^{n - m})$ \emph{simple}.
We denote by $[x]$ the total number of symbols $0$ and $1$ in $x$.
We color in green all simple inputs $(x, y)$ with $[x] + [y] = n$. The function $g$ is defined on such inputs.
We color in blue simple inputs $(x, y)$ such that $[x] + [y] < n$. The function $g$ is not defined on them (see Fig.~\ref{pic1}). 
The number of green inputs is  $(n + 1) \cdot 2^n$.
The number of blue inputs is $$\sum\limits_{i = 0}^{n - 1} \sum\limits_{j = 0}^{n - i - 1} 2^{i + j} = \sum\limits_{i = 0}^{n - 1} 2^i \cdot (2^{n - i} - 1) = n \cdot 2^n - (2^n - 1) = (n - 1) \cdot 2^n + 1$$

\newcommand{\good}{\cellcolor{green}}
\newcommand{\bad}{\cellcolor{blue}}
\begin{figure}
$$\begin{array}{c|c|c|c|c|c|c|c|c|c|c|c}
               & 0\ldots 0 & \ldots    & 1\ldots 1 & \rr 0\ldots 0 & \ldots & \rr 1\ldots 1 & \rr \rr 0\ldots 0 & \ldots & \rr \rr 1\ldots 1 & \ldots & \rr \ldots \rr  \\ \hline
     \rr \ldots \rr  & \good     & \good     & \good     & \bad       & \bad   & \bad       & \bad        & \bad   & \bad        & \bad   & \bad      \\ \hline
    0\rr \ldots \rr  &           &           &           & \good      & \good  & \good      & \bad        & \bad   & \bad        & \bad   & \bad      \\ \hline
    1\rr \ldots \rr  &           &           &           & \good      & \good  & \good      & \bad        & \bad   & \bad        & \bad   & \bad      \\ \hline
   00\rr \ldots \rr  &           &           &           &            &        &            & \good       & \good  & \good       & \bad   & \bad      \\ \hline
   01\rr \ldots \rr  &           &           &           &            &        &            & \good       & \good  & \good       & \bad   & \bad      \\ \hline
   10\rr \ldots \rr  &           &           &           &            &        &            & \good       & \good  & \good       & \bad   & \bad      \\ \hline
   11\rr \ldots \rr  &           &           &           &            &        &            & \good       & \good  & \good       & \bad   & \bad      \\ \hline
    \multicolumn{12}{c}{\ldots} \\ \hline
     0\ldots 0 &           &           &           &            &        &            &             &        &             &        & \good     \\ \hline
     \ldots    &           &           &           &            &        &            &             &        &             &        & \good     \\ \hline
     1\ldots 1 &           &           &           &            &        &            &             &        &             &        & \good     \\ \hline
\end{array}
$$
\caption{The part of matrix of function $g$ consisting of simple inputs.
}\label{pic1}
\end{figure}

Assume that classical complexity of $g$ is at most $c$.
Apply the above lemma to a half-duplex protocol of depth $c$ to compute $g$ and consider 
the corresponding partition into at most $L=2^c$ monochromatic rectangles.

Let $(x, y)$ and $(u, v)$ be different green inputs with $[x] = [u]$.
Then $g$ is defined on both $(x, v)$ and $(u, y)$ and
not all values $g(x, y)$, $g(x, v)$, $g(u, y)$, $g(u, v)$ are the same.
Thus, $(x, y)$ and $(u, v)$ are in different rectangles of the partition.

Note that, if a rectangle contains $k$ green inputs $(x_i, y_i)$, $[x_1] < [x_2] < \ldots < [x_k]$, then it also contains $\frac{k \cdot (k - 1)}{2}$ blue inputs $(x_i, y_j)$, $i < j$.
Let $k_i$ be the number of green inputs in the $i$th rectangle, then $i$th rectangle also contains at least $\frac{k_i \cdot (k_i - 1)}{2}$ blue inputs. Then $\sum\limits_{i = 1}^L k_i = (n + 1) \cdot 2^n$.
On the other hand, $\sum\limits_{i = 1}^L \frac{k_i (k_i - 1)}{2} \leq (n - 1) \cdot 2^n + 1$
as rectangles are disjoint. Re-writing the second inequality, we get
\begin{align*}
\sum\limits_{i = 1}^L( k_i^2 - k_i) &\leq (2n - 2)\cdot 2^n + 2\\
\sum\limits_{i = 1}^L k_i^2 &\leq (2n - 2)\cdot 2^n + 2 + (n + 1) \cdot 2^n\\
\sum\limits_{i = 1}^L k_i^2 &\leq (3n - 1)\cdot 2^n + 2.
\end{align*}
By Cauchy–Schwarz inequality, $\left(\sum\limits_{i=1}^L k_i\right)^2 \leq L \cdot \sum\limits_{i=1}^L k_i^2$. 
Therefore,
\begin{align*}
(n + 1)^2 \cdot 2^{2n} &\leq L \cdot ((3n - 1) \cdot 2^n + 2)\Rightarrow\\
L &\geq \frac{(n + 1)^2 \cdot 2^{2n}}{(3n - 1) \cdot 2^n + 2}\Rightarrow\\
L &\geq \frac{n \cdot 2^n}{3}\cdot\left(1 + o(1)\right)\Rightarrow\\
c=\log_2 L &\geq n + \log_2 n - \log_2 3 + o(1)\qed
\end{align*}
\renewcommand{\qed}{}
\end{proof}

So, the function $g$ is an example of a partial function with 
logarithmic gap between local classical complexity and half-duplex complexity with adversary.
We would also like to find an example of a total function with the same gap. 
Unfortunately, we could not find any such example.
In the rest of the paper we present 
an example of a total function with a constant gap --- our function has half-duplex complexity 5 and classical complexity 6.
We will start with a simpler example of a constant gap between classical complexity and half-duplex complexity with honest adversary.

\section{A separation of half-duplex complexity
  with honest adversary from classical complexity
for total functions}\label{s4}

\begin{theorem}
  There is a total function with
  classical   complexity   4
and half-duplex complexity with honest adversary at most  3.
\end{theorem}
\begin{proof}
  Let $X=Y=\{\rr,00,01,10,11\}$
and $Z=\{0,1,2\}$.
First we design a 3 round  half-duplex protocol,
computing a total function  $f:X\times Y\to Z$ (against honest adversary) and then we prove that its 
classical complexity
is at least 4.

In that protocol only the first round is not 
classical. In that round both players depending on their inputs
choose one of the three actions.
That action is determined by the first symbol of the input:
$\rr$ means ``receive'', and  0,1 mean  ``send 0,1'', respectively.
In the second round Bob sends a bit and 
Alice receives, and in the third round vice versa.
In those two rounds the players send the bit that was received in the first
round provided they chose to receive in the first round.
Otherwise players send the second bits
of their inputs.

The result of the
protocol is a function of bits sent (and received) in the second and third rounds
(we will call it \emph{the 2--3-transcript} in the sequel). Since those rounds are classical, that
guarantees that  Alice and Bob output the same result,
whatever happens in the first round.
Let us see what is the 2--3-transcript equal to.
If Bob's input is
$\rr$ and  Alice's input is $0,1$,
then 2--3-transcript coincides with Alice's input.
Similarly, if Bob's input is 0,1
and Alice's input is
$\rr$, then  2--3-transcript coincides with 
Bob's input reversed.
If both  Alice's and
Bob's inputs are
$0,1$, then  2--3-transcript is formed from 
the second bits of  Bob and
Alice. 
Finally, if  both  Alice's and
Bob's inputs are equal to  $\rr$, then  2--3-transcript consists of bits chosen
by the adversary. Since we assume honest adversary, it is equal either to 00, or to 11.
The  2--3-transcript is shown in the following table where  Alice's input indicates the row and 
 Bob's input the column:
$$
\begin{array}{|c|c|cc|cc|}\hline
&\rr&00 &01 &10 &11\\
\hline
\rr&00\text{ or } 11&00&10&01&11\\
\hline
00&00&00&10&00&10\\
01 &01&01&11&01&11\\
\hline
10&10&00&10&00&10\\
11&11&01&11&01&11\\
\hline
\end{array}
$$
We can see that if the output function
takes the same values on the arguments 00 and 11, then the
protocol
computes a total function.
We will define the output function as the  2--3-transcript 
with identified 00 and 11. Let us represent 2--3-transcripts
by natural numbers. In this way we get the following
function: 
$$
\begin{array}{c|ccccc|}
&\rr&00 &01 &10 &11\\
\hline
\rr&0&0&2&1&0\\
00&0&0&2&0&2\\
01 &1&1&0&1&0\\
10&2&0&2&0&2\\
11&0&1&0&1&0\\
\hline
\end{array}
$$
The key point is the following: since the adversary is honest, we need to identify
only 00 and 11. For a malicious adversary, we would need to identify all four 2--3-transcripts
and the function would become constant.

To show that the  classical   complexity of the
defined function is at least
4, we will find a fooling set of size $10>2^3$. 
The members of that set are colored magenta in 
the matrix:
$$
\left(
\begin{array}{ccccc}
0&0&\hm2&\hm1&\hm0\\
\hm0&0&2&0&\hm2\\
\hm1&1&0&1&0\\
\hm2&\hm0&2&0&2\\
0&\hm1&\hm0&1&0\\

\end{array}
\right)\qed 
$$
\renewcommand{\qed}{}
\end{proof}

Recently T.~Kuptsov~\cite{kuptsov} has shown that the half-duplex complexity with 
malicious adversary of this function is the same as its classical complexity, that is, is equal to 4.
Thus this function separates  half-duplex complexity with adversary from half-duplex complexity with honest adversary.
This also means that, to separate half-duplex complexity
with malicious adversary from classical complexity,
we have to look for another example. 

Since fooling sets have the direct product property, this example provides 
a sequence of functions with a linear gap between half-duplex complexity with honest adversary
and classical complexity. More specifically, let the function $f^n: X^n\times Y^n\to Z^n$ be the direct product 
of $n$ copies of $f$, that is, 
$f^n(x_1x_2\dots x_n,y_1y_2\dots y_n)=f(x_1,y_1)f(x_2,y_2)\dots  f(x_n,y_n)$.
Here $X,Y,Z,f$ are sets and the function from the proof of the previous theorem.
The arguments of  $f$ are length-$n$ strings over the alphabets $X,Y$ and
its value is a length-$n$ string over $Z$.

\begin{theorem}
 The classical   complexity   of $f^n$ is at least $\log_2(10)
 n$ while 
its  half-duplex complexity with honest adversary is at most  $3n$.
\end{theorem}
\begin{proof}
The second statement is obvious, since we can just 
apply $n$ times the protocol from the previous theorem to compute $n$ copies of $f$.

The lower bound for classical   complexity is proven as follows. Recall that the function $f$ has a fooling set $F$ of size 10.
We claim that its Cartesian power $F^n$ is a fooling set for $f^n$.

For the sake of contradiction assume that
that different pairs $(x,y)$ and $(a,b)$ are in $F^n$ and are in a monochromatic rectangle $R$.
Let $x=x_1\dots x_n$, $y=y_1\dots y_n$, $a=a_1\dots a_n$, $b=b_1\dots b_n$.
Then for some $i$ the pairs $(x_i,y_i)$ and  $(a_i,b_i)$ differ.
W.l.o.g. assume that $i=1$. Since $F$ is a fooling set for $f$,
the rectangle  $\{x_1,a_1\}\times\{y_1,b_1)$ is not monochromatic.
That is, this rectangle has two  pairs on which $f$ has different values. 

There are six possible cases.
Case 1: Assume first that $f(x_1,y_1)\ne f(a_1,b_1)$. Then $f(x,y)\ne f(a,b)$, which is a contradiction, as both pairs $(x,y), (a,b)$ are in $R$.
Case 2: Now assume that $f(x_1,y_1)\ne f(x_1,b_1)$.  Then $f(x,y)\ne f(x,b)$, which is again a contradiction, as both pairs $(x,y), (x,b)$ are in $R$.
The remaining cases are similar to these two cases.

Since $|F^n|=10^n$, the classical complexity of $f^n$ is at least $\log_2|F^n|=n\log_2(10)$.
\end{proof}

\section{A separation of half-duplex complexity
  with malicious adversary from classical complexity
for total functions}\label{s5}

\begin{theorem}\label{th-main}
  There is a total function
  with classical complexity 6
and half-duplex complexity (with malicious adversary) at most 5.
\end{theorem}

The proof of this theorem is similar to that of the previous one but is much more complicated.
The proof is divided into three sections.
In the first section we define a 5 round
half-duplex protocol $\Pi$ to compute 
a total function. 
In the second and third sections we prove
that the
classical complexity of the computed function is at least 6.

\subsection{The  half-duplex protocol}

\subsubsection{A general protocol}\label{ss-general}
Our general plan is the following.
The first round of the
protocol $\Pi$ is not
classical and all the remaining 4 rounds are  classical, in those rounds the players
recieve bits in the following sequence: Alice,
Bob, Alice, Bob.
The result of the 
protocol is a function of the sequence of bits sent in rounds 2--5
(that sequence will called the \emph{2--5-transcript}).
The inputs of the players consist of one or two
symbols.
The first symbol
of player's input is $\rr,0$ or 
$1$ and 
has the same meaning as before:
it indicates what to do in the first round.
The remaining 
symbol (if any)
indicates how to act in the
remaining 4 rounds. 

 More specifically,
 Bob's input is either $\rr$, or has the form $0\phi$ or  $1\phi$
 where the range of $\phi$ will be defined later.
 Alice's input is either $\rr\eta$, or has the form $0\psi$ or  $1\psi$
 where the ranges of $\eta,\psi$ will be defined later.
 The protocol is the following:
\begin{enumerate}
\item  If  Bob's input is 
  $\rr$, then in the first round 
  Bob receives (let $m$ denote the received bit), then he sends
  $m$ (the  bit he
  just received), then receives again, 
  then again he sends $m$ and then again receives.

\item If  Alice's input
  is $\rr\eta$, then she receives in the first round (let $m$ denote the received bit),
then she receives again (denote the received bit  by $i$),
then sends $m$ (the bit from the first round),
then again receives (let it be $k$), then again sends $m$ if $i=k$ and otherwise (if $i\ne k$)
in the last round  Alice sends a bit that depends on $\eta,i,k$ (how 
that bits depends on them, will be defined later).
\item If  Bob's input is $i\phi$, $i=0,1$,
then Bob sends $i$ in the first round, then he sends a bit that depends on $\phi$,
then receives  (let it be $j$), then sends a bit that depends on 
$j,\phi$, then again receives a bit.
\item  If Alice's input is 
  $j\psi$, $j=0,1$, then 
  Alice sends $j$ in the first round, then she receives
 (let it be $i$), then sends a bit that depends on
 $i,\psi$, then receives  (let it be  $k$) and finally sends a bit that depends on 
 $i,k,\psi$.
\end{enumerate}

According to this 
protocol, on the input pair 
$(\rr\eta, \rr)$ the events of the players are
the following
(assuming that the adversary sends $j$ to Alice and $i$ to Bob):
$$
\begin{array}{|c|c|c|c|c|c|}
\hline
	\text{Round} & 1 & 2 & 3 & 4 & 5 \\ \hline
	\text{Bob's event} & \text{``received $i$''} & \text{``sent $i$''} & \text{``received $j$''} & \text{``sent $i$''} & \text{``received $j$''} \\ \hline 
	\text{Alice's event} & \text{``received $j$''} & \text{``received $i$''} & \text{``sent $j$''} & \text{``received $i$''} & \text{``sent $j$''} \\ \hline 
\end{array}
$$
So far, this is only a general plan of protocol's construction. We have yet to 
define the ranges of $\eta,\phi,\psi$ and to define the  dependencies mentioned above.
However we can already see that, if the first round happens to be silent, then the
2--5-transcript of the protocol is one of these four sequences:
0000, 0101, 1010 and 1111. We will define Alice's and Bob's output
be 0, if the  2--5-transcript of the protocol is one of these four sequences,
and to the 2--5-transcript otherwise.
Such a definition guarantees that even for malicious adversary the computed function is total. 
We have identified only a quarter of 2--5-transcripts, which leaves an opportunity for the computed function be non-trivial.

\subsubsection{The first realization of the plan}
The simplest realization of this plan is the following: let $\eta\in\{0,1\}$ and 
$\phi,\psi\in\{0,1\}^2$ and let the protocol run as follows:
\begin{enumerate}
\item  If  Bob's input is 
  $\rr$, then in the first round 
  Bob receives (let $m$ denote the received bit), then he sends
  $m$, then receives again, 
  then again he sends $m$ and then again receives.

\item If  Alice's input
  is $\rr\eta$, then she receives in the first round (let $m$ denote the received bit),
then she receives again (let it be $i$),
then sends $m$,
then again receives (let it be $k$), then again sends $m$, if $i=k$, and $\eta$ otherwise.
\item If  Bob's input is $i\phi$, $i=0,1$, $\phi=\phi_1\phi_2$,
then Bob sends $i$ in the first round, then he sends $\phi_1$,
then receives, then sends $\phi_2$, then again receives.
\item  If Alice's input is 
  $j\psi$, $j=0,1$,  $\psi=\psi_1\psi_2$, then 
  Alice sends $j$ in the first round, then she receives, then sends  $\psi_1$, 
  then receives  and finally sends 
 $\psi_2$.
\end{enumerate}
As a result, Alice and  Bob produce the following 2--5 transcript
(each transcript is represented by a natural number using its binary expansion):
$$
\begin{array}{|c|c|cccc|cccc|}
\hline
\text{Bob's input} \to
&\rr &000&001&010&011 &100&101&110&111\\
\hline
\text{Alice's input}\downarrow&&&&&&&&&\\
\hline
\rr0 &\hr0,\hc{5},\hg{10},\hb{15} &\hr0&2&8&\hg{10} &\hc{5}&6&12&\hb{15}\\
\rr1 &\hr0,\hc{5},\hg{10},\hb{15} &\hr0&3&9&\hg{10} &\hc{5}&7&13&\hb{15}\\
\hline
000& \hr{0} &\hr{0}&2&8&\hg{10}    &\hr{0}&2&8&\hg{10}  \\ 
001&1   &1&3&9&11   &1&3&9&11    \\
010&4 &4&6&12&14 &4&6&12&14\\
011&\hc{5}   &\hc{5}&7&13 &\hb{15}    &\hc{5}&7&13 &\hb{15}    \\
\hline
100&\hg{10} &\hr{0}&2&8&\hg{10}    &\hr{0}&2&8&\hg{10}  \\ 
101&11   &1&3&9&11   &1&3&9&11    \\
110& 14 &4&6&12&14 &4&6&12&14\\
111&\hb{15}   &\hc{5}&7&13 &\hb{15}    &\hc{5}&7&13 &\hb{15}    \\
\hline
\end{array}
$$
Thus our protocol computes the function with  
the following matrix 
(zeros have colors indicating the respective 2--5 transcripts):
$$
\begin{array}{|c|cccc|cccc|}
\hline
0 &\hr0&2&8&\hg0 &\hc0&6&12&\hb0\\
0 &\hr0&3&9&\hg0 &\hc0&7&13&\hb0\\

\hline\hr0 &0&2&8&\hg0    &\hr0&2&8&\hg0  \\ 1   &1&3&9&11   &1&3&9&11    \\4 &4&6&12&14 &4&6&12&14\\
\hc0   &\hc0&7&13 &\hb0    &\hc0&7&13 &\hb0    \\

\hline\hg0 &\hr0&2&8&\hg0    &\hr0&2&8&\hg0  \\ 11   &1&3&9&11   &1&3&9&11    \\
14 &4&6&12&14 &4&6&12&14\\
\hb0   &\hc0&7&13 &\hb0    &\hc0&7&13 &\hb0    \\

\hline

\end{array}
$$
\renewcommand{\hr}[1]{{\color{red}#1}}
\renewcommand{\hb}[1]{{\color{blue}#1}}
\renewcommand{\hc}[1]{{\color{cyan}#1}}
\renewcommand{\hg}[1]{{\color{green}#1}}
By construction the half-duplex complexity of this function
is at most  5. Unfortunately, its classical  complexity is also at most 5. 
To show this, let us note that this matrix has 7 different columns (the 2nd column coincides with the 6th one and the 5th column with the 9th one). 
Each column of this matrix, except the first one, 
has at most  4 different numbers.
Therefore we can compute this 
function in 5 rounds using the following protocol.
First Bob lets Alice know the entire column, corresponding to his input.
If his input is not the first column,
then he sends 3 bits, and otherwise he sends only 2 bits.
In the first case 
Alice in two rounds sends Bob the value of the function. 
In the second case she does that in 3 rounds.

\subsubsection{The second (and final) realization of the plan}
To increase the classical complexity of this 
function, let us use most general dependencies in our 
protocol.
Namely, let $\phi$ range over all mappings that  
determine which bits sends 
Bob in the second and fourth rounds in all possible situations.
In other words, $\phi$ is a function with the domain 
$\{\Lambda,0,1\}$ and values   0,1 where  $\phi(\Lambda)$ is the bit sent in the second round and 
$\phi(j)$ is the bit sent in the fourth round after receiving  $j$ in the third round.
Thus Bob has
$1+2\times 8=17$ different inputs.

Similarly, 
$\eta$ ranges over all mappings that determine which bits
sends Alice in all possible situations (assuming that she receives in the first round).
The domain of $\eta$ is the set 
$\{001,101,010,110\}$ where $\eta(mik)$ is the bit which Alice sends
in the fifth round provided she receives  
$m,i,k$ in the first, second and fourth rounds, respectively.

The domain of $\psi$ is $\{0,1, 00, 01,10,11\}$ where
$\psi(i)$ is the bit Alice sends in the third round
provided she received $i$ in the second round
and
$\psi(ik)$ is the bit  Alice sends in the fifth round provided she received
$i,k$ in the second and fourth rounds, respectively.
Thus Alice has 
$2\times 2^6+2^4=144$  different inputs.

Now our half-duplex protocol is well defined.
By construction it computes a total function and we denote that function
by $U$. 
We will prove that 
the  classical communication 
complexity of $U$ is 6. 

The matrix of $U$
is very large, its size is $144\times17$, and therefore that proof cannot be visualized. 
Because of that, we found  two moderately small sub-functions of $U$
with the same classical complexity. By a sub-function
here we mean a restriction of $U$ on a set of the form
$X'\times Y'$ where $X'\subset X$, $Y'\subset Y$.
In terms of the matrix of the function,
this means deleting some rows and columns form the matrix. Note that
restricting the function cannot increase its 
half-duplex complexity, thus the  
half-duplex complexity of any sub-function of
$U$ is at most 5.

\subsubsection{Sub-functions of $U$ with communication complexity 6: the function $S$}
\newcommand{\five}{0}
\newcommand{\ten}{0}
\newcommand{\fifteenb}{0}
\newcommand{\tenb}{0}
\newcommand{\fifteen}{0}
\newcommand{\fivey}{0}
\newcommand{\fiveb}{0}
\newcommand{\teny}{0}
\newcommand{\fifteeny}{0}
\newcommand{\zeroy}{ 0}
\newcommand{\zerob}{0}
\newcommand{\zeror}{0}
\renewcommand{\hr}[1]{{#1}}
\renewcommand{\hg}[1]{{#1}}
\renewcommand{\hc}[1]{{#1}}
\renewcommand{\hb}[1]{{#1}}
\renewcommand{\hm}[1]{{#1}}

The smallest sub-function of $U$ with classical complexity 6 we managed to find
 has the matrix shown on Fig.~\ref{pic111}.
We denote this sub-function by $S$.
\begin{figure}
$$
\begin{array}{|c|cccccc|cccccc|}
\hline
0&0&0&2&9&0&0  &0 &7&7&12&12&0   \\
0&0&0&3&8&0&0 &0&7&7&13&13&0\\
0&0&0&2&8&0&0 &0&6&6&13&13&0\\\hline
0&0&7&7&13&13&0 &0&7&7&13&13&0\\
0&0 &6&6&12&12&14  &0 &6&6&12&12&14   \\
4&4 &6&6&12&12&14 &4 &6&6&12&12&14 \\
4&4 &6&6&9 &11&11 &4 &6&6&9 &11&11 \\
1&1&1&3&9&0&0 & 1&1&3&9&0&0 \\ 
0&0&0&2&8&0&0 & 0&0&2&8&0&0 \\ 
\hline
0&0&7&7&13&13&0 &0&7&7&13&13&0\\
14&0 &6&6&12&12&14  &0 &6&6&12&12&14   \\
14&4 &6&6&12&12&14 &4 &6&6&12&12&14 \\
11&4 &6&6&9 &11&11 &4 &6&6&9 &11&11 \\
0&1&1&3&9&0&0 & 1&1&3&9&0&0 \\ 
0&0&0&2&8&0&0 & 0&0&2&8&0&0 \\ 
\hline
\end{array}
$$
\caption{The matrix of the function $S$.
  Its classical  complexity is 6 and its half-duplex complexity is at most 5.}\label{pic111}
\end{figure}
Communication complexity of $S$ was found on a computer in 4 minutes using the dynamic programming. 
Since this computation cannot be performed by hand, we continued to look for a sub-function
of $U$ with a moderately small domain for which we can prove by hand the lower bound 6 for the
classical complexity. 

\subsubsection{Sub-functions of $U$ with communication complexity 6: the function $M$}
The matrix of the  function $M$ we found is shown on Fig.~\ref{pic2}.
\begin{figure}
$$
\begin{array}{|c|ccccccc|ccccccc|}
\hline0&\zeroy &\hg{2}&\zeroy&\hg{2}&\hg{8}&\hg{8}&\teny  &\fivey &\hg{6 }&\hg{6 }&\hg{12 }&\fifteeny &\hg{12}&\fifteeny\\ 

0&\zeroy &\hg{2}&\zeroy&\hg{2}&\hg{8}&\hg{8}&\teny  &\fivey &\hg{6 }&\hg{6 }&\hg{13 }&\fifteeny &\hg{13}&\fifteeny\\ 

0&\zeroy &\hg{2}&\zeroy&\hg{2}&\hg{9}&\hg{9}&\teny  &\fivey &\hg{7 }&\hg{7 }&\hg{12 }&\fifteeny &\hg{12}&\fifteeny\\

0&\zeroy &\hg{3}&\zeroy&\hg{3}&\hg{8}&\hg{8}&\teny  &\fivey &\hg{7 }&\hg{7 }&\hg{13 }&\fifteeny &\hg{13}&\fifteeny\\  

0&\zeroy &\hg{3}&\zeroy&\hg{3}&\hg{9}&\hg{9}&\teny  &\fivey &\hg{6 }&\hg{6 }&\hg{12 }&\fifteeny &\hg{12}&\fifteeny\\ 

0 &\zeroy&\hg{3}&\zeroy&\hg{3}&\hg{9}&\hg{9}&\teny  &\fivey &\hg{7 }&\hg{7 }&\hg{13 }&\fifteeny &\hg{13}&\fifteeny\\ 
\hline
0&0&2&0&2&8&8&\hr{11}&\hr{2}&\zeror&\hr{2}&\hr{8}&\hr{8}&\hr{11}&\hr{11} \\

0&0&3&0&3&\hr{12}&\hr{14}&\hr{12}&\hr{3}&\zeror&\hr{3}&12&\hr{14}&12&\hr{14}\\

\hg1&1&2&1&2&9&9&\hr{11}&\hr{2}&1&\hr{2}&9&9&\hr{11}&\hr{11}\\

\hg1&1&2&1&2&\hr{13}&\hr{14}&\hr{13} &\hr{2}&1&\hr{2}&13&\hr{14}&13&\hr{14}\\

\hg1&1&3&1&3&\hr{12}&\hr{14}&\hr{12}&\hr{3}&1&\hr{3}&12&\hr{14}&12&\hr{14}\\

\hg1&1&3&1&3&8&8&\hr{11} &\hr{3}&1&\hr{3}&\hr{8}&\hr{8}&\hr{11}&\hr{11}\\

\hg4  &4&4&\hr{6}&\hr{6}&9&9&\hr{11}   &4&6&6&9&9&\hr{11}&\hr{11}\\

\hg4  &4&4&\hr{7}&\hr{7}&8&8&\hr{11}   &4&7&7&\hr{8}&\hr{8}&\hr{11}&\hr{11}\\

\hg4  &4&4&\hr{6}&\hr{6}&\hr{12}&\hr{14}&\hr{12}   &4&6&6&12&\hr{14}&12&\hr{14}   \\

\hg4  &4&4&\hr{7}&\hr{7}&\hr{13}&\hr{14}&\hr{13}    &4&7&7&13&\hr{14}&13 &\hr{14}  \\

\fiveb  &\fiveb&\five&\hr{7}&\hr{7}&8&8&\hr{11}   &\fiveb&7&7&\hr{8}&\hr{8}&\hr{11}&11  \\
\hline

\hg{11} &0&2&0&2&9&9&11&\hr{2}&\zeror&\hr{2}&\hr{9}&\hr{9}&11&11\\

\hg{14}&0&3&0&3&\hr{13}&14&\hr{13}&\hr{3}&\zeror&\hr{3}&13&14&13&14\\

\hg{11}&\hr1&2&1&2&8&8&11&\hr{2}&\hr{1}&\hr{2}&\hr{8}&\hr{8}&11&\hr{11}\\

\hg{14}&\hr1&2&\hr{1}&2&\hr{12}&14&\hr{12}&\hr{2}&\hr{1}&\hr{2}&12&14&12&\hr{14}\\

\hg{11}&\hr1&3&\hr{1}&3&9&9&11 &\hr{3}&\hr{1}&\hr{3}&\hr{9}&\hr{9}&11&\hr{11}\\

\hg{14}&\hr1&3&\hr{1}&3&\hr{13}&14&\hr{13}&\hr{3}&\hr{1}&\hr{3}&13&14&13&\hr{14}\\

\hg{11} &\hr{4} &\hr{4}&\hr{6}&\hr{6}&8&8&11   &\hr{4}&6&6&\hr{8}&\hr{8}&11&\hr{11}\\

\hg{14} &\hr{4} &\hr{4}&\hr{6}&\hr{6}&\hr{13}&14&\hr{13}   &\hr{4}&6&6&13&14&13  &\hr{14}  \\

\hg{11}&\hr{4}  &\hr{4}&\hr{7}&\hr{7}&9&9&11   &\hr{4}&7&7&\hr{9}&\hr{9}&11&\hr{11}\\

\tenb &\hr{4} &\hr{4}&\hr{7}&\hr{7}&9&9&\tenb   &\hr{4}&7&7&\hr{9}&\hr{9}&\ten&\tenb\\

\fifteenb &\hr{4} &\hr{4}&\hr{6}&\hr{6}&\hr{12}&\fifteen&\hr{12}   &\hr{4}&6&6&12&\fifteenb&12 &\fifteenb  \\

\hg{14}&\hr{4}  &\hr{4}&\hr{7}&\hr{7}&\hr{12}&14&\hr{12}    &\hr{4}&7&7&12&14&12 &\hr{14} \\
\hline
\end{array}
$$
\caption{The matrix of the function $M$. Its classical
  communication complexity is  6, while its half-duplex complexity is at most 5.}\label{pic2}
\end{figure}
On Fig.~\ref{pic3} we indicated, on the left and on the top,
Alice's and Bob's inputs. Contemplating this matrix, we can imagine how the matrix of $U$ looks like. 
The domain of $M$ is the smallest sub-domain of  $U$
for which our method to prove the lower bound 6 works. Later we will explain this in more detail. 

Looking at Fig.~\ref{pic3} one can verify that each entry in the matrix is indeed the result
of the above described half-duplex protocol $\Pi$ on the corresponding input pair. Also one can verify that  matrices on Fig.~\ref{pic2}
and Fig.~\ref{pic3} are identical. This implies that 
the function with the matrix shown on Fig.~\ref{pic3} 
is a sub-function of $U$ and hence its half-duplex complexity is at most 5.
However both verifications are quite time consuming and therefore we will explain in the next section
how to verify faster that its half-duplex complexity is at most 5. 
\begin{figure}
$$
\begin{array}{|c|c|c|c|c|c|c|c|c|ccccccc|ccccccc|}
\hline
\multicolumn{6}{|c|}{}&\multicolumn{3}{|c|}{\text{Arg. of func. }\phi\downarrow} &\multicolumn{14}{|c|}{\text{The value of function }\phi\downarrow}\\
\hline
\multicolumn{6}{|c|}{}&\multicolumn{3}{|c|}{\Lambda}&0 &0&0&0&1&1&1&0&0&0&1&1&1&1\\
\multicolumn{6}{|c|}{}&\multicolumn{3}{|c|}{0 }              &0 &1&0&1&0&0&1&1&0&1&0&0&1&1\\
\multicolumn{6}{|c|}{}&\multicolumn{3}{|c|}{1 }              &0 &0&1&1&0&1&0&0&1&1&0&1&0&1\\
\hline
&&&&&&&\text{Inp.  B} \to&\rr
&\multicolumn{7}{|c|}{0\phi}&\multicolumn{7}{|c|}{1\phi}\\
\hline\text{A.}\eta&&&001&010&101&110& \text{Inp. A}\downarrow& &&&&&&&&&&&&&&\\  
\hline
&&&0&0&0&0&  & 0&\zeroy &\hg{2}&\zeroy&\hg{2}&\hg{8}&\hg{8}&\teny  &\fivey &\hg{6 }&\hg{6 }&\hg{12 }&\fifteeny &\hg{12}&\fifteeny\\ 
&&&0&0&0&1&  & 0&\zeroy &\hg{2}&\zeroy&\hg{2}&\hg{8}&\hg{8}&\teny  &\fivey &\hg{6 }&\hg{6 }&\hg{13 }&\fifteeny &\hg{13}&\fifteeny\\ 
&&&0&1&1&0&  & 0&\zeroy &\hg{2}&\zeroy&\hg{2}&\hg{9}&\hg{9}&\teny  &\fivey &\hg{7 }&\hg{7 }&\hg{12 }&\fifteeny &\hg{12}&\fifteeny\\  
 \text{V.}\eta
&&&1&0&1&1&\rr\eta & 0&\zeroy &\hg{3}&\zeroy&\hg{3}&\hg{8}&\hg{8}&\teny  &\fivey &\hg{7 }&\hg{7 }&\hg{13 }&\fifteeny &\hg{13}&\fifteeny\\  
&&&1&1&0&0& & 0&\zeroy &\hg{3}&\zeroy&\hg{3}&\hg{9}&\hg{9}&\teny  &\fivey &\hg{6 }&\hg{6 }&\hg{12 }&\fifteeny &\hg{12}&\fifteeny\\ 
&&&1&1&1&1& & 0 &\zeroy&\hg{3}&\zeroy&\hg{3}&\hg{9}&\hg{9}&\teny  &\fivey &\hg{7 }&\hg{7 }&\hg{13 }&\fifteeny &\hg{13}&\fifteeny\\ 
\hline
\text{A.}\psi &0&1&00&01&10&11&&
&&&&&&&&&&&&&&\\
\hline
&0&0&0&0&0&1&
&0&\hr0&2&0&2&8&8&\hr{11}&\hr{2}&\zeror&\hr{2}&\hr{8}&\hr{8}&\hr{11}&\hr{11} \\
&0&1&0&1&0&0&
&0&\hr0&3&0&3&\hr{12}&\hr{14}&\hr{12}&\hr{3}&\zeror&\hr{3}&12&\hr{14}&12&\hr{14}\\
&0&0&1&0&1&1&
&\hg1&1&2&1&2&9&9&\hr{11}&\hr{2}&1&\hr{2}&9&9&\hr{11}&\hr{11}\\
&0&1&1&0&1&0&
&\hg1&1&2&1&2&\hr{13}&\hr{14}&\hr{13} &\hr{2}&1&\hr{2}&13&\hr{14}&13&14\\
&0&1&\hr1&1&0&0&
&\hg1&1&3&1&3&\hr{12}&\hr{14}&\hr{12}&\hr{3}&1&\hr{3}&12&\hr{14}&12&\hr{14}\\
\text{V.} \psi
&0&0&1&1&0&1&0\psi
&\hg1&1&3&1&3&8&8&\hr{11} &\hr{3}&1&\hr{3}&\hr{8}&\hr{8}&\hr{11}&\hr{11}\\
&1&0&0&0&1&1&  
&\hg4  &4&4&\hr{6}&\hr{6}&9&9&\hr{11}   &4&6&6&9&9&\hr{11}&\hr{11}\\
&1&0&0&1&0&1&  
&\hg4  &4&4&\hr{7}&\hr{7}&8&8&\hr{11}   &4&7&7&\hr{8}&\hr{8}&\hr{11}&\hr{11}\\
&1&1&0&0&0&0&  
&\hg4  &4&4&\hr{6}&\hr{6}&\hr{12}&\hr{14}&\hr{12}   &4&6&6&12&\hr{14}&12&\hr{14}   \\
&1&1&0&1&1&0&  
&\hg4  &4&4&\hr{7}&\hr{7}&\hr{13}&\hr{14}&\hr{13}    &4&7&7&13&\hr{14}&13 &\hr{14}  \\
&1&0&1&1&0&1&  
&\fiveb  &\five&\five&\hr{7}&\hr{7}&8&8&\hr{11}   &\fiveb&7&7&\hr{8}&\hr{8}&\hr{11}&11   \\
\hline
&0&0&0&0&1&1&
&\hg{11} &0&2&0&2&9&9&11&\hr{2}&\zeror&\hr{2}&\hr{9}&\hr{9}&11&11\\
&0&1&0&1&1&0&
&\hg{14}&0&3&0&3&\hr{13}&14&\hr{13}&\hr{3}&\zeror&\hr{3}&13&14&13&14\\
&0&0&1&0&0&1&
&\hg{11}&\hr1&2&1&2&8&8&11&\hr{2}&\hr{1}&\hr{2}&\hr{8}&\hr{8}&11&\hr{11}\\
&0&1&1&0&0&0&
&\hg{14}&\hr1&2&\hr{1}&2&\hr{12}&14&\hr{12}&\hr{2}&\hr{1}&\hr{2}&12&14&12&\hr{14}\\
&0&0&1&1&1&1&
&\hg{11}&\hr1&3&\hr{1}&3&9&9&11 &\hr{3}&\hr{1}&\hr{3}&\hr{9}&\hr{9}&11&\hr{11}\\
\text{V.} \psi
&0&1&1&1&1&0&1\psi
&\hg{14}&\hr1&3&\hr{1}&3&\hr{13}&14&\hr{13}&\hr{3}&\hr{1}&\hr{3}&13&14&13&\hr{14}\\
&1&0&0&0&0&1&  
&\hg{11} &\hr{4} &\hr{4}&\hr{6}&\hr{6}&8&8&11   &\hr{4}&6&6&\hr{8}&\hr{8}&11&\hr{11}\\
&1&1&0&0&1&0&  &\hg{14} &\hr{4} &\hr{4}&\hr{6}&\hr{6}&\hr{13}&14&\hr{13}   &\hr{4}&6&6&13&14&13  &\hr{14}  \\
&1&0&0&1&1&1&  &\hg{11}&\hr{4}  &\hr{4}&\hr{7}&\hr{7}&9&9&11   &\hr{4}&7&7&\hr{9}&\hr{9}&11&\hr{11}\\
&1&0&0&1&1&0&  &\tenb &\hr{4} &\hr{4}&\hr{7}&\hr{7}&9&9&\tenb   &\hr{4}&7&7&\hr{9}&\hr{9}&\ten&\ten\\
&1&1&0&0&0&1& &\fifteenb &\hr{4} &\hr{4}&\hr{6}&\hr{6}&\hr{12}&\fifteen&\hr{12}   &\hr{4}&6&6&12&\fifteenb&12 &\fifteenb  \\
&1&1&0&1&0&0&  &\hg{14}&\hr{4}  &\hr{4}&\hr{7}&\hr{7}&\hr{12}&14&\hr{12}    &\hr{4}&7&7&12&14&12 &\hr{14} \\
\hline
\end{array}
$$
\caption{The matrix of function $M$ with an indication of Alice's and Bob's inputs.
``A.''  means ``Argument of'', ``V.''  means ``Value of'',
``Inp. A''  means ``Input of Alice'' and ``Inp. B''  means ``Input of Bob''. 
}\label{pic3}
\end{figure}

\subsubsection{A fast verification that the half-duplex complexity of  
the function with matrix on Fig.~\ref{pic2} is at most 5}

We exhibit the  partition the matrix from Fig.~\ref{pic2} into monochromatic rectangles  
corresponding to leaves of Alice's and Bob's trees (in the protocol $\Pi$).
Such a partition is shown on Fig.~\ref{pic2a}. The double lines partition the matrix
according to events happening in the first round.
For example, the bottom right rectangle consists of all input pairs for which both Alice and
 Bob send 1, the upper right rectangle consists of all input pairs for which Alice
 receives 1 from Bob  and the upper left rectangle of all input pairs for which both parties receive in the first
 round. These rectangles will be denoted by  $P_{\rr\rr}, P_{\rr0}, P_{\rr1},P_{0\rr}, P_{1\rr},P_{00}, P_{01}, P_{10}, P_{11}$.

In Fig.~\ref{pic2a}, we have restored the 2--5-transcripts, that is, 0, 5, 10 and 15 are not yet identified.
The only exception is the rectangle  $P_{\rr\rr}$, in which all  2--5-transcripts  0, 5, 10 and 15 can appear.
The number in each cell indicates in which part  
(0-part or 1-part) the respective cell falls in each round. Namely, its leading bit is equal to the bit sent
by Bob in the second round, the second bit   is equal to the bit sent
by Alice in the third  round, and so on. In the second round
all numbers are partitioned into small ones
(less than 8) and large ones (larger than or equal to 8).
This  partition is indicated by vertical lines. In the third round the numbers are partitioned
according to the second bit: red numbers have 0 and blue ones have 1.
In the fourth round the numbers are partitioned
according to the third bit, those with 1 as the third bit are shown in italic.  Finally,
in the fifth round the numbers are partitioned
into even and odd ones.
\renewcommand{\five}{5}
\renewcommand{\ten}{10}
\renewcommand{\fifteenb}{15}
\renewcommand{\tenb}{10}
\renewcommand{\fifteen}{15}
\renewcommand{\fivey}{5}
\renewcommand{\fiveb}{5}
\renewcommand{\teny}{10}
\renewcommand{\fifteeny}{10}
\renewcommand{\zeroy}{ 0}
\renewcommand{\zerob}{0}
\renewcommand{\zeror}{0}
\newcommand{\uu}[1]{{\color{red}#1}}
\renewcommand{\hg}[1]{{#1}}
\renewcommand{\hc}[1]{{#1}}
\newcommand{\vv}[1]{{\color{blue}#1}}
\renewcommand{\hr}[1]{{#1}}
\renewcommand{\hm}[1]{{#1}}
\newcommand{\nott}[1]{\it{#1}}

\begin{figure}
$$
\begin{array}{||c||cccc|ccc||ccc|cccc||}\hline
\hline
0&\uu{0} &\uu{\nott 2}&\uu{0}&\uu{\nott2}&\uu{8}&\uu{8}&\uu{\nott{10}}  &\vv5 &\vv{\nott6}&\vv{\nott6}&\vv{12}&\vv{\nott{15}} &\vv{12}&\vv{\nott{15}}\\ 

0&\uu{0} &\uu{\nott2}&\uu0&\uu{\nott2}&\uu{8}&\uu{8}&\uu{\nott{10}} &\vv5 &\vv{\nott6 }&\vv{\nott6 }&\vv{13 }&\vv{\nott{15}} &\vv{13}&\vv{\nott{15}}\\ 

0&\uu{0} &\uu{\nott2}&\uu{0}&\uu{\nott2}&\uu{9}&\uu{9}&\uu{\nott{10}}  &\vv5 &\vv{\nott7}&\vv{\nott7}&\vv{12}&\vv{\nott{15}} &\vv{12}&\vv{\nott{15}}\\

0&\uu{0} &\uu{\nott3}&\uu{0}&\uu{\nott3}&\uu{8}&\uu{8}&\uu{\nott{10}}  &\vv5 &\vv{\nott7}&\vv{\nott7}&\vv{13}&\vv{\nott{15}} &\vv{13}&\vv{\nott{15}}\\  

0&\uu{0} &\uu{\nott3}&\uu{0}&\uu{\nott3}&\uu{9}&\uu{9}&\uu{\nott{10}}  &\vv5 &\vv{\nott6}&\vv{\nott6}&\vv{12}&\vv{\nott{15}} &\vv{12}&\vv{\nott{15}}\\ 

0&\uu{0}&\uu{\nott3}&\uu{0}&\uu{\nott3}&\uu{9}&\uu{9}&\uu{\nott{10}}  &\vv5 &\vv{\nott7}&\vv{\nott7}&\vv{13}&\vv{\nott{15}} &\vv{13}&\vv{\nott{15}}\\ 

\hline\hline
\uu0&\uu0&\uu{\nott2}&\uu0&\uu{\nott2}&\uu8&\uu8&\uu{\nott{11}}&\uu{\nott2}&\uu0&\uu{\nott2}&\uu{8}&\uu{8}&\uu{\nott{11}}&\uu{\nott{11}} \\

\uu0&\uu0&\uu{\nott3}&\uu0&\uu{\nott3}&\vv{12}&\vv{\nott{14}}&\vv{12}&\uu{\nott3}&\uu0&\uu{\nott3}&\vv{12}&\vv{\nott{14}}&\vv{12}&\vv{\nott{14}}\\

\uu1&\uu1&\uu{\nott2}&\uu1&\uu{\nott2}&\uu9&\uu9&\uu{\nott{11}}&\uu{\nott2}&\uu1&\uu{\nott2}&\uu9&\uu9&\uu{\nott{11}}&\uu{\nott{11}}\\

\uu1&\uu1&\uu{\nott2}&\uu1&\uu{\nott2}&\vv{13}&\vv{\nott{14}}&\vv{13} &\uu{\nott2}&\uu1&\uu{\nott2}&\vv{13}&\vv{\nott{14}}&\vv{13}&\vv{\nott{14}}\\

\uu1&\uu1&\uu{\nott3}&\uu1&\uu{\nott3}&\vv{12}&\vv{\nott{14}}&\vv{12}&\uu{\nott3}&\uu1&\uu{\nott3}&\vv{12}&\vv{\nott{14}}&\vv{12}&\vv{\nott{14}}\\

\uu1&\uu1&\uu{\nott3}&\uu1&\uu{\nott3}&\uu8&\uu8&\uu{\nott{11}} &\uu{\nott3}&\uu1&\uu{\nott3}&\uu{8}&\uu{8}&\uu{\nott{11}}&\uu{\nott{11}}\\

\vv{4}  &\vv4&\vv4&\vv{\nott6}&\vv{\nott6}&\uu9&\uu9&\uu{\nott{11}}   &\vv4&\vv{\nott6}&\vv{\nott6}&\uu9&\uu9&\uu{\nott{11}}&\uu{\nott{11}}\\

\vv{4}  &\vv4&\vv4&\vv{\nott7}&\vv{\nott7}&\uu8&\uu8&\uu{\nott{11}}   &\vv4&\vv{\nott7}&\vv{\nott7}&\uu{8}&\uu{8}&\uu{\nott{11}}&\uu{\nott{11}}\\

\vv{4}  &\vv4&\vv4&\vv{\nott6}&\vv{\nott6}&\vv{12}&\vv{\nott{14}}&\vv{12}   &\vv4&\vv{\nott6}&\vv{\nott6}&\vv{12}&\vv{\nott{14}}&\vv{12}&\vv{\nott{14}}   \\

\vv{4}  &\vv4&\vv4&\vv{\nott7}&\vv{\nott7}&\vv{13}&\vv{\nott{14}}&\vv{13}    &\vv4&\vv{\nott7}&\vv{\nott7}&\vv{13}&\vv{\nott{14}}&\vv{13} &\vv{\nott{14}}  \\

\vv5  &\vv5&\vv5&\vv{\nott7}&\vv{\nott7}&\uu8&\uu8&\uu{\nott{11}}   &\vv5&\vv{\nott7}&\vv{\nott7}&\uu{8}&\uu{8}&\uu{\nott{11}}&\uu{\nott{11}}  \\
\hline\hline

\uu{\nott{11}} &\uu0&\uu{\nott2}&\uu0&\uu{\nott2}&\uu9&\uu9&\uu{\nott{11}}&\uu{\nott2}&\uu0&\uu{\nott2}&\uu{9}&\uu{9}&\uu{\nott{11}}&\uu{\nott{11}}\\

\vv{\nott{14}}&\uu0&\uu{\nott3}&\uu0&\uu{\nott3}&\vv{13}&\vv{\nott{14}}&\vv{13}&\uu{\nott3}&\uu0&\uu{\nott3}&\vv{13}&\vv{\nott{14}}&\vv{13}&\vv{\nott{14}}\\

\uu{\nott{11}}&\uu1&\uu{\nott2}&\uu1&\uu{\nott2}&\uu8&\uu8&\uu{\nott{11}}&\uu{\nott2}&\uu1&\uu{\nott2}&\uu{8}&\uu{8}&\uu{\nott{11}}&\uu{\nott{11}}\\

\vv{\nott{14}}&\uu1&\uu{\nott2}&\uu{1}&\uu{\nott2}&\vv{12}&\vv{\nott{14}}&\vv{12}&\uu{\nott2}&\uu1&\uu{\nott2}&\vv{12}&\vv{\nott{14}}&\vv{12}&\vv{\nott{14}}\\

\uu{\nott{11}}&\uu1&\uu{\nott3}&\uu{1}&\uu{\nott3}&\uu9&\uu9&\uu{\nott{11}} &\uu{\nott3}&\uu1&\uu{\nott3}&\uu{9}&\uu{9}&\uu{\nott{11}}&\uu{\nott{11}}\\

\vv{\nott{14}}&\uu1&\uu{\nott3}&\uu{1}&\uu{\nott3}&\vv{13}&\vv{\nott{14}}&\vv{13}&\uu{\nott3}&\uu1&\uu{\nott3}&\vv{13}&\vv{\nott{14}}&\vv{13}&\vv{\nott{14}}\\

\uu{\nott{11}} &\vv4 &\vv4&\vv{\nott6}&\vv{\nott6}&\uu8&\uu8&\uu{\nott{11}}   &\vv4&\vv{\nott6}&\vv{\nott6}&\uu{8}&\uu{8}&\uu{\nott{11}}&\uu{\nott{11}}\\

\vv{\nott{14}} &\vv4 &\vv4&\vv{\nott6}&\vv{\nott6}&\vv{13}&\vv{\nott{14}}&\vv{13}   &\vv4&\vv{\nott6}&\vv{\nott6}&\vv{13}&\vv{\nott{14}}&\vv{13}  &\vv{\nott{14}}  \\

\uu{\nott{11}}&\vv4  &\vv4&\vv{\nott7}&\vv{\nott7}&\uu9&\uu9&\uu{\nott{11}}   &\vv4&\vv{\nott7}&\vv{\nott7}&\uu{9}&\uu{9}&\uu{\nott{11}}&\uu{\nott{11}}\\

\uu{\nott{10}} &\vv4 &\vv4&\vv{\nott7}&\vv{\nott7}&\uu9&\uu9&\uu{\nott{10}}   &\vv4&\vv{\nott7}&\vv{\nott7}&\uu{9}&\uu{9}&\uu{\nott{10}}&\uu{\nott{10}}\\

\vv{\nott{15}}&\vv4 &\vv4&\vv{\nott6}&\vv{\nott6}&\vv{12}&\vv{\nott{15}}&\vv{12}   &\vv4&\vv{\nott6}&\vv{\nott6}&\vv{12}&\vv{\nott{15}}&\vv{12} &\vv{\nott{15}} \\

\vv{\nott{14}}&\vv4  &\vv4&\vv{\nott7}&\vv{\nott7}&\vv{12}&\vv{\nott{14}}&\vv{12}    &\vv4&\vv{\nott7}&\vv{\nott7}&\vv{12}&\vv{\nott{14}}&\vv{12} &\vv{\nott{14}} \\

\hline\hline

\end{array}
$$
\caption{Partition of the matrix of
  $M$ into monochromatic rectangles by the protocol $\Pi$.}\label{pic2a}
\end{figure}
It is not hard to verify that this partition corresponds 
to the  protocol described in Section~\ref{ss-general}. It is also easy to check that 
all the rectangles in the partition are monochromatic.
For example, all the 5 rectangles consisting of large red slanted odd numbers 
are monochromatic: all their entries are equal to 11.

\subsection{The lower bound of communication complexity of $U$ and $M$.}

It remains to prove that the 
classical complexity of functions
$U,M$ is at least 6. In the proof for $M$ it will be easier 
to forget that $M$ is a sub-function of $U$ and instead define $M$
as the function whose matrix is shown on Fig.~\ref{pic2}.

\begin{theorem}\label{th3} The
  classical  communication  complexity of  $U,M$ is at least  6.
\end{theorem}

Unfortunately, we cannot prove this using partitions into
monochromatic rectangles, as both matrices can be partitioned into
$30\le2^5$ monochromatic rectangles. These rectangles correspond to the 
leaves of classical protocols performed by Alice and 
Bob in  2--5 rounds. We have $9\times 16$ such leaves 
(9 events in the first rounds are multiplied by the number of leaves in a 
classical protocol of depth 4). However some of these
rectangles can be joined together
and after that we obtain  30 monochromatic rectangles.
Let us show this partition for the matrix of the function
$S$
(for $U$ and $M$ the partition is similar). 
For any number $i\ne 0$ the part of the matrix consisting of entries $i$ can
by partitioned into 2 rectangles. For $i=1$ and
$i=2$ the partition is shown in Fig.~\ref{pic4}.
\begin{figure}
\renewcommand{\hg}[1]{{\color{green}#1}}
\renewcommand{\hc}[1]{{\color{cyan}#1}}
\renewcommand{\hb}[1]{{\color{blue}#1}}
\renewcommand{\hy}[1]{{\color{red}#1}}
$$
\begin{array}{|c|cccccc|cccccc|}
\hline0&0&0&\hg2&9&0&0  &0 &7&7&12&12&0   \\
0&0&0&3&8&0&0 &0&7&7&13&13&0\\
0&0&0&\hg2&8&0&0 &0&6&6&13&13&0\\
\hline
0&0&7&7&13&13&0 &0&7&7&13&13&0\\
0&0 &6&6&12&12&14  &0 &6&6&12&12&14   \\
4&4 &6&6&12&12&14 &4 &6&6&12&12&14 \\
4&4 &6&6&9 &11&11 &4 &6&6&9 &11&11 \\
\hc1&\hy1&\hy1&3&9&0&0 & \hy1&\hy1&3&9&0&0 \\ 
0&0&0&\hb2&8&0&0 & 0&0&\hb2&8&0&0 \\ 
\hline
0&0&7&7&13&13&0 &0&7&7&13&13&0\\
14&0 &6&6&12&12&14  &0 &6&6&12&12&14   \\
14&4 &6&6&12&12&14 &4 &6&6&12&12&14 \\
11&4 &6&6&9 &11&11 &4 &6&6&9 &11&11 \\
0&\hy1&\hy1&3&9&0&0 & \hy1&\hy1&3&9&0&0 \\ 
0&0&0&\hb2&8&0&0 & 0&0&\hb2&8&0&0 \\ 
\hline
\end{array}
$$
\caption{Partition of 1s and 2s in
  the matrix of  $S$ into two rectangles. Each rectangle has its own color.}\label{pic4}
\end{figure}
And 0s can be partitions into  6 rectangles as shown in Fig.~\ref{pic5}.
\begin{figure}
\renewcommand{\hr}[1]{{\color{red}#1}}
\renewcommand{\hg}[1]{{\color{green}#1}}
\renewcommand{\hc}[1]{{\color{cyan}#1}}
\renewcommand{\hb}[1]{{\color{blue}#1}}
\renewcommand{\hy}[1]{{\color{yellow}#1}}
\renewcommand{\hm}[1]{{\color{magenta}#1}}
$$
\begin{array}{|c|cccccc|cccccc|}
\hline\hm0&\hy0&\hy0&2&9&\hy0&\hy0  &\hy0 &7&7&12&12&\hy0   \\
\hm0&\hy0&\hy0&3&8&\hy0&\hy0 &\hy0&7&7&13&13&\hy0\\
\hm0&\hy0&\hy0&2&8&\hy0&\hy0 &\hy0&6&6&13&13&\hy0\\
\hline
\hm0&\hc0&7&7&13&13&\hr0 &\hc0&7&7&13&13&\hr0\\
\hm0&\hc0 &6&6&12&12&14  &\hc0 &6&6&12&12&14   \\
4&4 &6&6&12&12&14 &4 &6&6&12&12&14 \\
4&4 &6&6&9 &11&11 &4 &6&6&9 &11&11 \\
1&1&1&3&9&\hb0&\hb0 & 1&1&3&9&\hb0&\hb0 \\ 
\hm0&\hg0&\hg0&2&8&\hb0&\hb0 & \hg0&\hg0&2&8&\hb0&\hb0 \\ 
\hline
\hm0&\hc0&7&7&13&13&\hr0 &\hc0&7&7&13&13&\hr0\\
14&\hc0 &6&6&12&12&14  &\hc0 &6&6&12&12&14   \\
14&4 &6&6&12&12&14 &4 &6&6&12&12&14 \\
11&4 &6&6&9 &11&11 &4 &6&6&9 &11&11 \\
\hm0&1&1&3&9&\hb0&\hb0 & 1&1&3&9&\hb0&\hb0 \\ 
\hm0&\hg0&\hg0&2&8&\hb0&\hb0 & \hg0&\hg0&2&8&\hb0&\hb0 \\ 
\hline

\end{array}
$$
\caption{Partition of zeros in the
  matrix of $S$ into 6 rectangles. Each rectangle has its own color.}\label{pic5}
\end{figure}
This number (30) of monochromatic rectangles cannot be decreased, which can be shown by fooling sets.

The proofs for functions  $U$ and
$M$ are almost 
identical. More specifically, some lemmas that are proved analytically for 
$U$ can by proved for $M$ using pictures.
We will show that 
for every vertical partition of the matrix into two parts
(that is, columns are divided into two parts) at least one part has a fooling set of size 17
(hence its communication complexity is at least 5), and the similar statement
 holds for horizontal partitions.

\subsubsection{The proof for horizontal partitions}

We distinguish in the matrix 25 rectangles denoted by
$$
R_0,R_1,\dots,R_4,R_6,\dots,R_9,R_{11},\dots,R_{14},S_1,\dots,S_4,S_6,\dots,S_9,S_{11},\dots S_{14}
$$
($R_5,R_{10},R_{15},S_0,S_5,S_{10},S_{15}$ are skipped). Those rectangles are called
\emph{fooling rectangles}.
All the entries of rectangles $R_i,S_i$ are equal to $i$.
Besides, for any two cells  $u,v$ from different rectangles
the set $\{u,v\}$ is a fooling set for the matrix, that is,
no monochromatic rectangle includes both $u$ and $v$.
In other words, if we pick any one cell from every
fooling rectangle, the resulting set of cells is a fooling set.
The existence of such a set of rectangles proves
only that the matrix has a fooling set of size 25.
However fooling rectangles have the following feature:
\emph{for every division of rows into two parts,
the rows of one of these parts intersect at least 
17 fooling rectangles}. That part thus has a fooling set of size 17.

We will define now fooling rectangles. For the matrix of 
$M$ the reader can skip the definition and look at  Fig.~\ref{pic7} instead, where fooling rectangles are indicated by colors.
\begin{figure}
\renewcommand{\hr}[1]{{\color{red}#1}}
\renewcommand{\hg}[1]{{\color{blue}#1}}
\renewcommand{\hc}[1]{{#1}}
\renewcommand{\hb}[1]{{#1}}
\renewcommand{\hm}[1]{{\color{red}#1}}
\renewcommand{\hy}[1]{{#1}}
\renewcommand{\fifteenb}{0}
\renewcommand{\tenb}{0}
\renewcommand{\ten}{\hc0}
\renewcommand{\fifteen}{\hb0}
\renewcommand{\fivey}{\hy0}
\renewcommand{\fiveb}{0}
\renewcommand{\five}{0}
\renewcommand{\teny}{\hy 0}
\renewcommand{\fifteeny}{\hy 0}
\renewcommand{\zeroy}{\hy 0}
\renewcommand{\zerob}{\hb0}
\renewcommand{\zeror}{\hm0}
$$
\begin{array}{|c|ccccccc|ccccccc|}
\hline
0&\zeroy &\hg{2}&\zeroy&\hg{2}&\hg{8}&\hg{8}&\teny  &\fivey &\hg{6 }&\hg{6 }&\hg{12 }&\fifteeny &\hg{12}&\fifteeny\\ 
0&\zeroy &\hg{2}&\zeroy&\hg{2}&\hg{8}&\hg{8}&\teny  &\fivey &\hg{6 }&\hg{6 }&\hg{13 }&\fifteeny &\hg{13}&\fifteeny\\ 
0&\zeroy &\hg{2}&\zeroy&\hg{2}&\hg{9}&\hg{9}&\teny  &\fivey &\hg{7 }&\hg{7 }&\hg{12 }&\fifteeny &\hg{12}&\fifteeny\\  
0&\zeroy &\hg{3}&\zeroy&\hg{3}&\hg{8}&\hg{8}&\teny  &\fivey &\hg{7 }&\hg{7 }&\hg{13 }&\fifteeny &\hg{13}&\fifteeny\\  
0&\zeroy &\hg{3}&\zeroy&\hg{3}&\hg{9}&\hg{9}&\teny  &\fivey &\hg{6 }&\hg{6 }&\hg{12 }&\fifteeny &\hg{12}&\fifteeny\\ 
0 &\zeroy&\hg{3}&\zeroy&\hg{3}&\hg{9}&\hg{9}&\teny  &\fivey &\hg{7 }&\hg{7 }&\hg{13 }&\fifteeny &\hg{13}&\fifteeny\\ 
\hline
0&0&2&0&2&8&8&\hr{11}&\hr{2}&\zeror&\hr{2}&\hr{8}&\hr{8}&\hr{11}&\hr{11} \\
0&0&3&0&3&\hr{12}&\hr{14}&\hr{12}&\hr{3}&\zeror&\hr{3}&12&\hr{14}&12&\hr{14}\\
\hg1&1&2&1&2&9&9&\hr{11}&\hr{2}&1&\hr{2}&\hr9&\hr9&\hr{11}&\hr{11}\\
\hg1&1&2&1&2&\hr{13}&\hr{14}&\hr{13} &\hr{2}&1&\hr{2}&13&\hr{14}&13&\hr{14}\\
\hg1&1&3&1&3&\hr{12}&\hr{14}&\hr{12}&\hr{3}&1&\hr{3}&12&\hr{14}&12&\hr{14}\\
\hg1&1&3&1&3&8&8&\hr{11} &\hr{3}&1&\hr{3}&\hr{8}&\hr{8}&\hr{11}&\hr{11}\\
\hg4  &4&4&\hr{6}&\hr{6}&9&9&\hr{11}   &4&6&6&\hr9&\hr9&\hr{11}&\hr{11}\\
\hg4  &4&4&\hr{7}&\hr{7}&8&8&\hr{11}   &4&7&7&\hr{8}&\hr{8}&\hr{11}&\hr{11}\\
\hg4  &4&4&\hr{6}&\hr{6}&\hr{12}&\hr{14}&\hr{12}   &4&6&6&12&\hr{14}&12&\hr{14}   \\
\hg4  &4&4&\hr{7}&\hr{7}&\hr{13}&\hr{14}&\hr{13}    &4&7&7&13&\hr{14}&13 &\hr{14}  \\
\fiveb  &\fiveb&\five&\hr{7}&\hr{7}&8&8&\hr{11}   &\fiveb&7&7&\hr{8}&\hr{8}&\hr{11} &\hr{11}   \\
\hline
\hg{11} &0&2&0&2&9&9&11&\hr{2}&\zeror&\hr{2}&\hr{9}&\hr{9}&11&11\\
\hg{14}&0&3&0&3&\hr{13}&14&\hr{13}&\hr{3}&\zeror&\hr{3}&13&14&13&14\\
\hg{11}&\hr1&2&1&2&8&8&11&\hr{2}&\hr{1}&\hr{2}&\hr{8}&\hr{8}&11&11\\
\hg{14}&\hr1&2&\hr{1}&2&\hr{12}&14&\hr{12}&\hr{2}&\hr{1}&\hr{2}&12&14&12&14\\
\hg{11}&\hr1&3&\hr{1}&3&9&9&11 &\hr{3}&\hr{1}&\hr{3}&\hr{9}&\hr{9}&11&11\\
\hg{14}&\hr1&3&\hr{1}&3&\hr{13}&14&\hr{13}&\hr{3}&\hr{1}&\hr{3}&13&14&13&14\\
\hg{11} &\hr{4} &\hr{4}&\hr{6}&\hr{6}&8&8&11   &\hr{4}&6&6&\hr{8}&\hr{8}&11&11\\
\hg{14} &\hr{4} &\hr{4}&\hr{6}&\hr{6}&\hr{13}&14&\hr{13}   &\hr{4}&6&6&13&14&13  &14  \\
\hg{11}&\hr{4}  &\hr{4}&\hr{7}&\hr{7}&9&9&11   &\hr{4}&7&7&\hr{9}&\hr{9}&11&11\\
\tenb &\hr{4} &\hr{4}&\hr{7}&\hr{7}&9&9&\tenb   &\hr{4}&7&7&\hr{9}&\hr{9}&\ten&\tenb\\
\fifteenb &\hr{4} &\hr{4}&\hr{6}&\hr{6}&\hr{12}&\fifteen&\hr{12}   &\hr{4}&6&6&12&\fifteenb&12 &\fifteenb  \\
\hg{14}&\hr{4}  &\hr{4}&\hr{7}&\hr{7}&\hr{12}&14&\hr{12}    &\hr{4}&7&7&12&14&12 &14 \\
\hline
\end{array}
$$
\caption{Fooling rectangles for $M$. The rectangle $R_i$ consists of red numbers
$i$. The rectangle $S_i$ consists of blue numbers $i$.}\label{pic7}
\end{figure}

\begin{definition}[fooling rectangles] 
  Recall that the value of the 
  function is essentially equal to the 2--5-transcript of the computation,
  which is a 4-bit sequence 
$abcd$.
  Therefore it is convenient, in the definition of
  $R_i,S_i$, to consider $i$ as a 4-bit sequence $abcd$. We first define rectangles
 $R_i$. 
If  $a\ne c$, then 
 \begin{align*}
R_{abcd}=R_{ab\bar a d}&= \{*\psi \mid \psi(a)=b,  \psi(a\bar a)=d\}
\times
\{\bar b\phi\mid \phi(\Lambda)=a,\phi(b)=\bar a\}.
 \end{align*}
Here $*$ denotes any bit 0,1 and 
 $\bar b=1-b$.  If $a=c$ and $b\ne d$, then
  \begin{align*}
R_{abcd}&=R_{aba\bar b}=\{\bar a\psi \mid \psi(a)=b,  \psi(aa)=\bar b\}
\times
\{*\phi\mid \phi(\Lambda)=a, \phi(b)=a\}
 \end{align*}
  Finally, if  $a=c,b=d$, which happens only if  $a=c=b=d=0$, then
\begin{align*}
R_{abcd}=R_{0000}=R_0&= \{*\psi \mid \psi(0)=0,  \psi(00)=0\}
\times
\{*\phi\mid \phi(\Lambda)=0,\phi(0)=0\}.
\end{align*}
  We now define rectangles $S_i$. 
If $a\ne c$, then 
 \begin{align*}
S_{abcd}=S_{ab\bar ad}&= \{\rr\eta \mid \eta(bac)=d\}
\times
\{b\phi\mid \phi(\Lambda)=a,\phi(b)=\bar a\}.
 \end{align*}
 Otherwise, if $a=c$, then 
 \begin{align*}
S_{abcd}&= S_{aba\bar b}= 
\{a\psi\mid \psi(a)=b,\psi(aa)=\bar a\}\times\{\rr\}.
 \end{align*}
 The first equality holds, since rectangles 
 $S_{0000},S_{0101},S_{1010},S_{1111}$ are not defined. 
For reader's convenience, explicit definitions of fooling rectangles follow:
\begin{align*}
R_0&=
\{*\psi\mid  \psi(0)=0,  \psi(00)=0\}
\times\{1\phi  \mid \phi(\Lambda)=0,  \phi(0)=0\}\\
R_1&= \{1\psi  \mid \psi(0)=0,  \psi(00)=1\}
\times\{*\phi\mid  \phi(\Lambda)=\phi(0)=0\}\\
R_4&= \{1\psi  \mid \psi(0)=1,  \psi(00)=0\}
 \times \{*\phi\mid  \phi(\Lambda)=\phi(1)=0\}\\
 R_{11}&= \{0\psi   \mid \psi(1)=0,  \psi(11)=1\}
 \times \{*\phi\mid  \phi(\Lambda)=\phi(0)=1\}\\
 R_{14}&=\{0\psi   \mid \psi(1)=1,  \psi(11)=0\}
 \times \{*\phi\mid  \phi(\Lambda)=\phi(1)=1\}\\
 R_2&= \{*\psi \mid  \psi(0)=0,  \psi(01)=0\}
\times\{1\phi \mid \phi(\Lambda)=0, \phi(0)=1\}\\ 
R_3&= \{*\psi \mid  \psi(0)=0,  \psi(01)=1\}
\times\{1\phi \mid \phi(\Lambda)=0, \phi(0)=1\}\\ 
R_6&= \{*\psi \mid  \psi(0)=1,  \psi(01)=0\}
\times\{0\phi \mid \phi(\Lambda)=0, \phi(1)=1\}\\ 
R_7&= \{*\psi \mid  \psi(0)=1,  \psi(01)=1\}
\times\{0\phi \mid \phi(\Lambda)=0, \phi(1)=1\}\\ 
R_8&= \{*\psi \mid  \psi(1)=0,  \psi(10)=0\}
\times\{1\phi \mid \phi(\Lambda)=1, \phi(0)=0\}\\ 
R_9&= \{*\psi \mid  \psi(1)=0,  \psi(10)=1\}
\times\{1\phi \mid \phi(\Lambda)=1, \phi(0)=0\}\\ 
R_{12}&= \{*\psi \mid  \psi(1)=1,  \psi(11)=0\}
\times\{0\phi \mid \phi(\Lambda)=1, \phi(1)=0\}\\ 
R_{13}&=\{*\psi \mid  \psi(1)=1,  \psi(11)=1\}
\times\{0\phi \mid \phi(\Lambda)=1, \phi(1)=0\}\\ 
S_1&=\{0\psi   \mid \psi(0)=0,  \psi(00)=1\}\times \{r\},\\
S_4&=\{0\psi   \mid \psi(0)=1,  \psi(00)=0\}\times\{r\},\\
S_{11}&=\{1\psi  \mid \psi(1)=0,  \psi(11)=1\}\times\{r\},\\
S_{14}&=\{1\psi  \mid \psi(1)=1,  \psi(11)=0\}\times\{r\},\\
S_2&=\{\rr\eta \mid \eta(001)=0\}\times\{0\phi \mid \phi(\Lambda)=0,\phi(0)=1\},\\
S_3&=\{\rr\eta \mid \eta(001)=1\}\times\{0\phi \mid \phi(\Lambda)=0,\phi(0)=1\},\\
S_6&=\{\rr\eta \mid \eta(011)=0\}\times\{1\phi \mid \phi(\Lambda)=0,\phi(1)=1\},\\
S_7&=\{\rr\eta \mid \eta(011)=1\}\times\{1\phi \mid \phi(\Lambda)=0,\phi(1)=1\},\\
S_{8}&=\{\rr\eta \mid \eta(100)=0\}\times\{0\phi \mid \phi(\Lambda)=1,\phi(0)=0\},\\
S_{9}&=\{\rr\eta \mid \eta(100)=1\}\times\{0\phi \mid \phi(\Lambda)=1,\phi(0)=0\},\\
S_{12}&=\{\rr\eta \mid \eta(110)=0\}\times\{1\phi \mid \phi(\Lambda)=1,\phi(1)=0\},\\
S_{13}&=\{\rr\eta \mid \eta(110)=1\}\times\{1\phi \mid \phi(\Lambda)=1,\phi(1)=0\},\\
\end{align*}

\end{definition}


\begin{lemma}\label{l1}
(1) All cells of rectangles  $R_i,S_i$ include the number $i$ only.

  (2) If $u,v$ are any cells from different fooling rectangles, then the set
 $\{u,v\}$ cannot be covered by a monochromatic rectangle
(in particular, $u\ne v$, that is, fooling rectangles are pair wise disjoint).
\end{lemma}
\begin{proof}
  For the function  $M$ this can be verified directly by examining Fig.~\ref{pic7}.
For the function $U$ the proof is in the Appendix. 
\end{proof}

We call two fooling rectangles \emph{horizontally adjacent}, if their first projections
intersect (there is a row that intersects both rectangles).
Consider the non-oriented graph whose nodes are fooling rectangles and
edges connect horizontally adjacent rectangles.
We say that a vertex $u$ is a \emph{neighbor} of a vertex $v$ if $u,v$ are adjacent.
In particular, each node is its own neighbor.
It turns out that this graph has the following expansion property:
\begin{lemma}\label{l2}
  Every set of 9 vertices of the graph has at
  least 17 neighbors.
\end{lemma}
The proof of this lemma is deferred to Appendix.
Let us derive from this lemma the proof for horizontal partitions. 
Assume that rows of the matrix are partitioned in sets  $W$ and $V$.
Let us denote by $\mathcal R$ the set of all fooling rectangles which
intersect some row from $W$.
If there are at least  17 such rectangles, then we are done. Otherwise there are at least 
 $25-16=9$ rectangles outside $\mathcal R$. For those rectangles every row intersecting the rectangle belongs to $V$.
Thus every neighbor $y$ of every vertex $x$ outside $\mathcal R$ intersects a row in $V$. (Indeed there is a row $s$
that intersects both  $x$ and $y$. That row is not in $W$, as otherwise
$x$ would be in  $\mathcal R$.  Hence $s\in V$ and $y$ intersects a row in  $V$.)
Since there are at least 9 rectangles outside 
$\mathcal R$, the lemma implies that there are at least 
17 fooling rectangles intersecting a row in $V$.  
 The case of horizontal partitions is completed.

Now we are able to explain how the set Alice's inputs was reduced when
we restricted the function $U$ to get the 
function $M$. Each row in the matrix of $U$ yields a clique in the graph,
the set of edges of the graph is the union of those cliques. We have chosen a minimal set of cliques
whose union covers all the edges of the graph and removed the rows corresponding to the 
remaining cliques. 

\subsubsection{Proof for vertical partitions}
Let us call two fooling  rectangles \emph{vertically adjacent}
if their second projections intersect.
Unfortunately, the analog of Lemma~\ref{l2} is not true anymore.
Indeed, the connected components of the vertical adjacency graph are
the following: 
\begin{align*}
\{S_1,S_{4},S_{11},S_{14}\},\\
\{R_0,S_2,S_{3},S_6,S_{7},R_{1},R_{2},R_{3},R_{4},R_6,R_7\},\\
\{S_8,S_{9},S_{12},S_{13},R_{8},R_{9},R_{11},R_{12},R_{13},R_{14}\}.
\end{align*}
These components correspond to three sub-matrices
into which the matrix $M$ is divided by ordinary vertical lines in Fig.~\ref{pic2a}.
Let us join the first and the second connected
components. In this way we obtain
a vertical partition of the matrix into two sub-matrices, where the first sub-matrix
intersects 15 fooling rectangles,
and the second one intersects 10 fooling rectangles. 

This obstacle forces to increase the number of fooling rectangles.
This can be done by adding rectangles with cells containing 0.
To do this, we will reduce the rectangle $R_0$ and add  rectangles called $R_5,R_{10},R_{15},S_0$.
Namely, let 
\begin{align*}
R_{abab}&= \{*\psi \mid \psi(a)=\psi(aa)=b,\psi(\bar a)\ne\psi(\bar a\bar a)\}\\
\times&
\{\bar b\phi\mid \phi(\Lambda)=\phi(b)=a,\phi(\bar b)=\bar a\},\text{ where }a,b=0,1,\nonumber \\
 S_{0}&=\{\rr\eta \mid \eta \text{ arbitrary }\}\times\{i\phi\mid\phi(\Lambda)=\phi(i)\}
 \end{align*}
Here
are the explicit definitions of rectangles  $R_0,R_5,R_{10},R_{15}$:
\begin{align*}
R_0&=
\{*\psi\mid  \psi(0)=0,  \psi(00)=0,  \psi(1)\ne\psi(11)\}
\times\{1\phi  \mid \phi(\Lambda)=0,  \phi(0)=0,\phi(1)=1\}\\
R_5&=
\{*\psi\mid  \psi(0)=1,  \psi(00)=1,  \psi(1)\ne\psi(11)\}
\times\{0\phi \mid \phi(\Lambda)=0,  \phi(1)=0,\phi(0)=1\},\\
R_{10}&=
\{*\psi\mid  \psi(1)=0,  \psi(11)=0,  \psi(0)\ne\psi(00)\}
\times\{1\phi  \mid \phi(\Lambda)=1,  \phi(0)=1,\phi(1)=0\}\\
R_{15}&=
\{*\psi\mid  \psi(1)=1,  \psi(11)=1,  \psi(0)\ne\psi(00)\}
\times\{0\phi  \mid \phi(\Lambda)=1,  \phi(1)=1,\phi(0)=0\}.
\end{align*}

In fact, for the function $M$, the new version of the rectangle $R_0$ coincides with the old one, since the difference between
them is due to those rows that were removed from the matrix of $U$.
The fooling rectangles for the matrix of $M$ are shown in Fig.~\ref{pic8}.
\begin{figure}
\renewcommand{\hr}[1]{{\color{red}#1}}
\renewcommand{\hg}[1]{{\color{blue}#1}}
\renewcommand{\hc}[1]{{\color{cyan}#1}}
\renewcommand{\hb}[1]{{\color{green}#1}}
\renewcommand{\hm}[1]{{\color{yellow}#1}}
\renewcommand{\hy}[1]{{\color{magenta}#1}}
\renewcommand{\fifteenb}{0}
\renewcommand{\tenb}{0}
\renewcommand{\ten}{\hc0}
\renewcommand{\fifteen}{\hb0}
\renewcommand{\fivey}{\hy0}
\renewcommand{\fiveb}{0}
\renewcommand{\five}{\hm0}
\renewcommand{\teny}{\hy 0}
\renewcommand{\fifteeny}{\hy 0}
\renewcommand{\zeroy}{\hy 0}
\renewcommand{\zerob}{\hb0}
\renewcommand{\zeror}{\hr0}
$$
\begin{array}{|c|ccccccc|ccccccc|}
\hline
0&\zeroy &\hg{2}&\zeroy&\hg{2}&\hg{8}&\hg{8}&\teny  &\fivey &\hg{6 }&\hg{6 }&\hg{12 }&\fifteeny &\hg{12}&\fifteeny\\ 
0&\zeroy &\hg{2}&\zeroy&\hg{2}&\hg{8}&\hg{8}&\teny  &\fivey &\hg{6 }&\hg{6 }&\hg{13 }&\fifteeny &\hg{13}&\fifteeny\\ 
0&\zeroy &\hg{2}&\zeroy&\hg{2}&\hg{9}&\hg{9}&\teny  &\fivey &\hg{7 }&\hg{7 }&\hg{12 }&\fifteeny &\hg{12}&\fifteeny\\  
0&\zeroy &\hg{3}&\zeroy&\hg{3}&\hg{8}&\hg{8}&\teny  &\fivey &\hg{7 }&\hg{7 }&\hg{13 }&\fifteeny &\hg{13}&\fifteeny\\  
0&\zeroy &\hg{3}&\zeroy&\hg{3}&\hg{9}&\hg{9}&\teny  &\fivey &\hg{6 }&\hg{6 }&\hg{12 }&\fifteeny &\hg{12}&\fifteeny\\ 
0 &\zeroy&\hg{3}&\zeroy&\hg{3}&\hg{9}&\hg{9}&\teny  &\fivey &\hg{7 }&\hg{7 }&\hg{13 }&\fifteeny &\hg{13}&\fifteeny\\ 
\hline
0&0&2&0&2&8&8&\hr{11}&\hr{2}&\zeror&\hr{2}&\hr{8}&\hr{8}&\hr{11}&\hr{11} \\
0&0&3&0&3&\hr{12}&\hr{14}&\hr{12}&\hr{3}&\zeror&\hr{3}&12&\hr{14}&12&\hr{14}\\
\hg1&1&2&1&2&9&9&\hr{11}&\hr{2}&1&\hr{2}&\hr9&\hr9&\hr{11}&\hr{11}\\
\hg1&1&2&1&2&\hr{13}&\hr{14}&\hr{13} &\hr{2}&1&\hr{2}&13&\hr{14}&13&\hr{14}\\
\hg1&1&3&1&3&\hr{12}&\hr{14}&\hr{12}&\hr{3}&1&\hr{3}&12&\hr{14}&12&\hr{14}\\
\hg1&1&3&1&3&8&8&\hr{11} &\hr{3}&1&\hr{3}&\hr{8}&\hr{8}&\hr{11}&\hr{11}\\
\hg4  &4&4&\hr{6}&\hr{6}&9&9&\hr{11}   &4&6&6&\hr9&\hr9&\hr{11}&\hr{11}\\
\hg4  &4&4&\hr{7}&\hr{7}&8&8&\hr{11}   &4&7&7&\hr{8}&\hr{8}&\hr{11}&\hr{11}\\
\hg4  &4&4&\hr{6}&\hr{6}&\hr{12}&\hr{14}&\hr{12}   &4&6&6&12&\hr{14}&12&\hr{14}   \\
\hg4  &4&4&\hr{7}&\hr{7}&\hr{13}&\hr{14}&\hr{13}    &4&7&7&13&\hr{14}&13 &\hr{11}  \\
\fiveb  &\fiveb&\five&\hr{7}&\hr{7}&8&8&\hr{11}   &\fiveb&7&7&\hr{8}&\hr{8}&\hr{11} &\hr{11}   \\
\hline
\hg{11} &0&2&0&2&9&9&11&\hr{2}&\zeror&\hr{2}&\hr{9}&\hr{9}&11&11\\
\hg{14}&0&3&0&3&\hr{13}&14&\hr{13}&\hr{3}&\zeror&\hr{3}&13&14&13&14\\
\hg{11}&\hr1&2&1&2&8&8&11&\hr{2}&\hr{1}&\hr{2}&\hr{8}&\hr{8}&11&11\\
\hg{14}&\hr1&2&\hr{1}&2&\hr{12}&14&\hr{12}&\hr{2}&\hr{1}&\hr{2}&12&14&12&14\\
\hg{11}&\hr1&3&\hr{1}&3&9&9&11 &\hr{3}&\hr{1}&\hr{3}&\hr{9}&\hr{9}&11&11\\
\hg{14}&\hr1&3&\hr{1}&3&\hr{13}&14&\hr{13}&\hr{3}&\hr{1}&\hr{3}&13&14&13&14\\
\hg{11} &\hr{4} &\hr{4}&\hr{6}&\hr{6}&8&8&11   &\hr{4}&6&6&\hr{8}&\hr{8}&11&11\\
\hg{14} &\hr{4} &\hr{4}&\hr{6}&\hr{6}&\hr{13}&14&\hr{13}   &\hr{4}&6&6&13&14&13  &14  \\
\hg{11}&\hr{4}  &\hr{4}&\hr{7}&\hr{7}&9&9&11   &\hr{4}&7&7&\hr{9}&\hr{9}&11&11\\
\tenb &\hr{4} &\hr{4}&\hr{7}&\hr{7}&9&9&\tenb   &\hr{4}&7&7&\hr{9}&\hr{9}&\ten&\tenb\\
\fifteenb &\hr{4} &\hr{4}&\hr{6}&\hr{6}&\hr{12}&\fifteen&\hr{12}   &\hr{4}&6&6&12&\fifteenb&12 &\fifteenb  \\
\hg{14}&\hr{4}  &\hr{4}&\hr{7}&\hr{7}&\hr{12}&14&\hr{12}    &\hr{4}&7&7&12&14&12 &14 \\
\hline
\end{array}
$$
\caption{Fooling rectangles for vertical partitions for 
the matrix of $M$. The rectangle $R_i$ for $i\ne 10,15,20$
consists of red
numbers $i$. The rectangle $S_i$ for $i\ne0$ consists of blue
numbers $i$.
The rectangles $R_5,R_{10},R_{15}$ consist of only one cell each: yellow, \hc{cyan}
and green zeros. The rectangle $S_0$ consists of \hy{magenta} zeros.}\label{pic8}
\end{figure}

We need an analogue of the Lemma ~\ref{l1} for new fooling rectangles:
\begin{lemma}\label{l3}
(1) All cells of the rectangles $R_0,R_5,R_{10},R_{15},S_0$
contain $0$.

(2) Any two cells $u,v$
from different rectangles from the list $R_0,S_0,R_5,R_{10},R_{15}$ cannot be covered
by one monochromatic rectangle (the cells of old fooling 
rectangles are marked with non-zeros, so they can also be added to this list).
\end{lemma}
For $M$, this lemma is verified directly by looking at Fig.~\ref{pic8}. For $U$, the lemma is proved in Appendix.

To complete the proof of the theorem, we need an analog of 
Lemma~\ref{l2} for the vertical adjacency graph. The number of fooling rectangles has increased to 29,
so instead of 9 we now  have $29-16=13$.
\begin{lemma}\label{l5}
Any set of 13 vertices of the vertical adjacency graph has at least 17 neighbors.
\end{lemma}
This lemma is proven in Appendix. 
As before, it implies that for any verical division of the matrices of $M$ and $U$ in two parts
at least one of the parts has a fooling set of 17 cells.
 
The case of vertical  partitions is completed (modulo Lemmas~\ref{l3} and~\ref{l5}). 
Theorems~\ref{th-main} and~\ref{th3} are proved.

\section{Open questions}
1. Is there a partial function $f$ on words
of length $n$ over a fixed alphabet with a super-logarithmic
gap between the local half-duplex and the local classical complexities?

2. Is there a total function with values 0,1 for which the
half-duplex complexity is strictly less than the classical complexity?

3. Is there a total function $f$ on words of length $n$ over a fixed alphabet with a super-constant
gap between the half-duplex and the classical complexity?

We believe that all our questions  answer in positive. Regarding  the third question, we think that the appropriate example is
the $n$th power of the function of Section 5 but we do not see any means to prove that.

\section{Acknowledgments}
The authors are sincerely grateful to Timur Kuptsov and anonymous referees for helpful suggestions.

\appendix
\section{Appendix: deferred proofs}
\begin{proof}[Proof of Lemma~\ref{l1} for $U$]
(1) This is verified directly.
Let us perform this verification, say, for the rectangle $R_6=R_{0110}$:
in the first round, both players send bits that are lost
and have no effect on the 2--5-transcript. In the second round
Bob sends 0 (because $\phi(\Lambda)=0$). Alice then
sends 1 (since $\psi(0)=1$). Then Bob sends 1 (since $\phi(1)=1$).
Finally, in the last round Alice sends 0, since $\psi(01)=0$.
We get the 2--5-transcript $0110=6$.

(2) Let $u,v$ belong to different fooling rectangles.

\emph{Case 1.} These rectangles have different subscripts.
Then $U(u)\ne U(v)$ due to item (1). Hence $u$ and $v$ cannot be covered by a monochromatic rectangle.

\emph{Case 2.} Rectangles containing $u,v$ have the same subscript.

Assume first that $u\in R_{abcd}$ and $v\in S_{abcd}$ where $a\ne c$.
Then $u=(*\psi,\bar b\phi)$, $v=(\rr\eta,b\phi')$.
Consider the input pair $(\rr\eta,\bar b\phi)$.
First,  the 2--5-transcript $t'$ for this pair is different from $t=abcd$.
Indeed, since for this input pair in the first round Alice receives
$\bar b$ from Bob, she sends in the third
round $\bar b$, and not $b$. 
Second, the values of the output function $p$ on $t$ and $t'$ are different,
since otherwise in both $t,t'$ the first and third bits coincide,
which is not the case for $t$.

Assume now that $u\in R_{abcd}$ and $v\in S_{abcd}$ where  $a=c$ and $b\ne d$.
Then $u=(\bar a\psi,*\phi)$, $v=(a\psi',\rr)$.
Consider the input pair  $(\bar a\psi,\rr)$.
First, the transcript $t'$ on this pair is different from $abcd$.
Indeed, since in the first round Bob receives $\bar a$ from Alice, he sends in the second round
$\bar a$. Therefore, $t'$ begins with $\bar a$ and hence 
differs from $t=abcd$. And again the values of the output function on $t$ and $t'$ differ 
because $b\ne d$.
\end{proof}

\begin{proof}[Proof of Lemma~\ref{l2}]
We need an explicit description of edges of the graph.
\begin{lemma}\label{l2bis}
  The graph of horizontally adjacent fooling rectangles
  is a union of the following graphs
(see Fig.~\ref{pi-1}):
\begin{figure}[t]
\begin{center}
\includegraphics[scale=1]{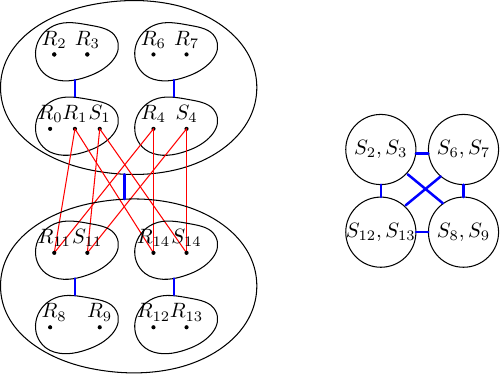}
\end{center}
\caption{The graph of horizontally adjacent fooling rectangles.
Each oval (circle) encircles vertices from one part of a bipartite or many-partite
graph. Blue lines between ovals represent edges connecting nodes from different parts of those bipartite graphs.
For instance, each blue line from the right represents 4
  edges. Red lines represent edges that are excluded from the complete 
 bipartite graph.}\label{pi-1}
\end{figure}
\begin{enumerate}
\item the complete 4-partite graph whose first part is  $\{S_2,S_3\}$, the second part is $\{S_{6},S_{7}\}$, the third part is $\{S_{8},S_{9}\}$
  and the fourth part is 
  $\{S_{12},S_{13}\}$.
\item the complete bipartite graph whose first part is $\{R_0,R_1,S_1\}$  and second part is
$\{R_2,R_3\}$
\item the complete bipartite graph whose first part is $\{R_4,S_4\}$ and second part is $\{R_6,R_7\}$
\item  the complete bipartite graph whose first part is $\{R_{11},S_{11}\}$ and second part is  $\{R_8,R_9\}$
\item  the complete bipartite graph whose first part is $\{R_{14},S_{14}\}$
and second part is $\{R_{12},R_{13}\}$
\item the complete bipartite graph whose first part is $\{R_0,R_1,R_2,R_3,R_4,R_{6},R_{7},S_1,S_4\}$
and second part is $\{R_{8},R_{9},R_{11},R_{12},R_{13},R_{14},S_{11},S_{14}\}$, except all edges of the complete bipartite graph whose first part is
$\{R_1,R_4,S_{11},S_{14}\}$ and second part is $\{R_{11},R_{14},S_1,S_4\}$. 
  \end{enumerate}
  \end{lemma}
\begin{proof}
Actually, 
we only need to show that all listed above edges are indeed present in the graph.
For the function $M$ this can be verified by contemplating Fig.~\ref{pic7}.
For the function $U$ this is proved as follows. 

Recall that two rectangles are connected by an edge if their first projections intersect.
We will call the second character of a word from the first projection of a rectangle its 
\emph{$\psi$- or $\eta$-projection}.

In the first item, we listed all the fooling rectangles whose first projections
contains words of the form $\rr*$. Their first projections
have the form $\{\rr\eta\mid\eta(bac)=d\}$. Such a set does not intersect
another  set $\{\rr\eta\mid\eta(\tb\ta\tc)=\td\}$  of this form only if
$a=\ta,b=\tb,c=\tc$, but $d\ne\td$. That is,
among the rectangles $S_2,S_3,S_6,S_7,S_8,S_9,S_{12},S_{13}$
there are no edges only between pairs
$(S_{0010},S_{0011})$, $(S_{0110},S_{0111})$,
$(S_{1000},S_{1001})$, $(S_{1100},S_{1101})$.

The first projections of all the other fooling rectangles
contain words of the form $*\psi$. Their first character,
except for rectangles $R_1,R_4,R_{11},R_{14},S_1,S_4,S_{11},S_{14}$,
can be arbitrary. For these exceptional rectangles, the first character is 1 for rectangles
$R_1,R_4,S_{11},S_{14}$ and is 0 for rectangles
  $R_{11},R_{14},S_1,S_4$. This explains the exclusion of edges in the sixth item.
It is not hard to verify that,
in items  2--5, no edges are excluded  for this reason: in these
items, in one of the parts, there are no rectangles from the list
$R_1,R_4,R_{11},R_{14},S_1,S_4,S_{11},S_{14}$. It remains to verify that in items 2--6
$\psi$- or $\eta$-projections intersect.

The edges in items 2--5 connect
the rectangles
$R_{abad}$ and $S_{abad}$
with rectangles of the form $R_{ab\bar a*}$.
We need to prove that $\psi$-projections of rectangles
$R_{abad}$ and $S_{abad}$ intersect with $\psi$-projections of rectangles of the form $R_{ab\bar a*}$.
Note that $\psi$-projections of $R_{abad}$ and $S_{abad}$ coincide and are equal to
$$
A=\{\psi \mid \psi(a)=b, \psi(aa)=d\}.
$$
On the other hand, the $\psi$-projection of  $R_{ab\bar ae}$ is equal to 
$$
C=C_e=\{\psi \mid \psi(a)=b, \psi(a\bar a)=e\}.
$$
Therefore the intersection
$A\cap C$ includes all $\psi$, with $\psi(a)=b,\psi(aa)=d, \psi(a\bar a)=e$. It is easy to see that
these requirements are compatible.

In the sixth item, the rectangles in the first part have the form
$R_{0bcd}$, $S_{0b0\bar b}$,
and in the second part rectangles have the same form, only 0 and 1 are swapped:
$R_{1b'c' d'}$, $S_{1b'1\bar{b'}}$.
Their $\psi$-projections are equal, respectively, to
\begin{align*}
\{\psi\mid \psi(0)=b,\psi(0c)=d\},\\
\{\psi\mid \psi(0)=b,\psi(00)=\bar b\},\\
\{\psi\mid \psi(1)=b',\psi(1c')=d'\},\\
\{\psi\mid \psi(1)=b',\psi(11)=\bar{b'}\}.
\end{align*}
Obviously, the first and second sets intersect with the third and fourth sets.
\end{proof}

Now we are able to prove Lemma~\ref{l2} by considering several 
cases.
The following sets of 9 vertices have
minimal number of neighbors
(17):
\begin{itemize}
\item The set of nodes in the top left oval. All 
  the vertices from both left ovals
  are their neighbors,  17 vertices in total.
\item The set of nodes from four circles on the right (8 nodes) plus
  any node from the top left oval which is incident to a red edge (e.g. $R_1$).
The extra node has 9 neighbors thus we get 17 nodes in total.
\end{itemize}
Let us show that this are
indeed the worst case. Assume that a set $A$ of vertices
contains $k$ vertices from the left component of connectivity and $l$ vertices from the right component, $k+l=9$.
Since there are only 8 vertices in the right component, we know that $k\ge1$.
We distinguish  three cases: $l=0$,  $l\ge 2$ and $l=1$.

\emph{Case $k=9,l=0$.} It is not difficult to verify that any vertex $u$ from the left component 
has at least 9 neighbors. So the number of non-neighbors of $u$ in the left component does not exceed
$17-9=8$. Therefore, $u$ is a neighbor of any set of $9$ vertices from
the left component. Since this is true for any vertex $u$, every node from the 17 nodes of the left component
is a neighbor of each  set of 9 vertices from the left component,
and we are done.

\emph{Case $l\ge 2$.}
Then one left vertex of $A$ has 9 or more neighbors in  the left component,
and the two right vertices of $A$ have 8 neighbors from the right component
(if these vertices are from the same circle, then they are neighbors of themselves,
otherwise they are neighbors of each other, and the remaining 6 vertices
will be their neighbors, since they belong to another part 
comparing to one of these two vertices).
In total, we get 17 neighbors.

\emph{Case $k=8,l=1$}. One $A$'s vertex on the right has 7 neighbors and it is enough to verify that
any set $B$ of 8 vertices of the left component has at least 10 neighbors.
Note that there are only four vertices in the left component that have 9 neighbors,
all the other have at least ten neighbors. Therefore, $B$ contains a vertex
with 10 neighbors.
\end{proof}

\begin{proof}[Proof of the Lemma~\ref{l3} for the function $U$]
(1) For rectangles $R_i$, this is verified directly.
For the rectangle $S_0$:
let the pair $(\rr\eta, i\phi)$ belong to $S_0$.
Let  $j$ denote $\phi(\Lambda)=\phi(i)$.
In the first round, Alice receives $i$ from Bob.
Therefore, Alice sends $i$ in the third 
round.
Thus, the transcript in 2--4 rounds looks like this: $jij$.
Finally, in the fifth round, Alice sends $i$ again, since
the bits she received in the second and fourth rounds coincide.
We get the transcript $jiji$ hence  $U(\rr\eta ,i\phi)=0$.

(2) Let $u,v$ belong to different fooling rectangles.
First, let us assume that both rectangles have the form $R_{abab}$. Say the first one
has the subscript $abab$,
and the second the subscript $a'b'a'b'$,
Let $u=(*\psi,\bar b\phi)$, $v=(*\psi',\bar{b'}\phi')$.
Then we know that
\begin{align}
\psi(a)=\psi(aa)=b \label{eq3}\\
\psi(\bar a)\ne \psi(\bar a\bar a) \label{eq4}\\
\phi'(\Lambda)=\phi'(b')=a' \label{eq5}\\
\phi'(\bar{b'})\ne a'. \label{eq6}
\end{align}

Consider the input pair $(*\psi,\bar{b'}\phi')$. We claim that $U(*\psi,\bar{b'}\phi')\ne 0$.
For the sake of contradiction, assume that the 2--5-transcript on this input pair
is a multiple of 5, that is, it has the form $\alpha\beta\alpha\beta$.
Then we additionally know that
\begin{align}
\psi(\alpha)=\psi(\alpha\alpha)=\beta \label{eq7}\\
\phi'(\Lambda)=\phi'(\beta)=\alpha. \label{eq8}
\end{align}
We claim that from~\eqref{eq3}--\eqref{eq8} it follows that $a=a'=\alpha$
and $b=b'=\beta$, and therefore the rectangles $R_{abab},R_{a'b'a'b'}$ coincide.
From~\eqref{eq7} and~\eqref{eq4} it follows that $\alpha\ne\bar a$, that is,
$\alpha=a$.
On the other hand, it follows from~\eqref{eq5} and~\eqref{eq8} that
$\alpha=\phi'(\Lambda)=a'$.
From~\eqref{eq8}
it follows that $\phi'(\beta)=\alpha=a$. Now from~\eqref{eq6} and~\eqref{eq8}
it follows that $\phi'(\bar{b'})\ne a'=\alpha=\phi'(\beta)$,
which means that $\bar{b'}\ne\beta$, that is, $\beta=b'$. Finally, from~\eqref{eq7} and~\eqref{eq3}
we conclude that $\beta=\psi(\alpha)=\psi(a)=b$.

It remains to consider the case when one rectangle is $S_0$,
and the other one  has the form $R_{abab}$. In this case we have
$u=(\rr\eta,i\phi)$, $v=(a'\psi',\bar{b'}\phi')$
and we know the equalities~\eqref{eq5} and~\eqref{eq6}.
We want to prove that $U(\rr\eta,\bar{b'}\phi')\ne0$.
For the sake of contradiction, assume that
the 2--5-transcript for  this input pair
has the form $\alpha\beta\alpha\beta$.
Then additionally we know ~\eqref{eq8},
and besides $\beta=\bar{b'}$, since Alice sends the bit
$\bar{b'}$ received from Bob in the third round.
From~\eqref{eq8} we derive that $\alpha=\phi'(\beta)=\phi'( \bar{b'})$.
And on the other hand, from~\eqref{eq5} and~\eqref{eq8}
it follows that $\alpha=\phi'(\Lambda)=a'$. Therefore,
$a'=\phi'( \bar{b'})$, which contradicts~\eqref{eq6}.
\end{proof}

\begin{proof}[Proof of Lemma~\ref{l5}]
Again, we need a more explicit description of edges of the graph.
\begin{lemma}\label{l4}
The vertical adjacency graph is equal to the graph
shown in Fig.~\ref{pi2} (for both functions $U,M$).
\end{lemma}
\begin{figure}[ht]
\begin{center}
\includegraphics[scale=1]{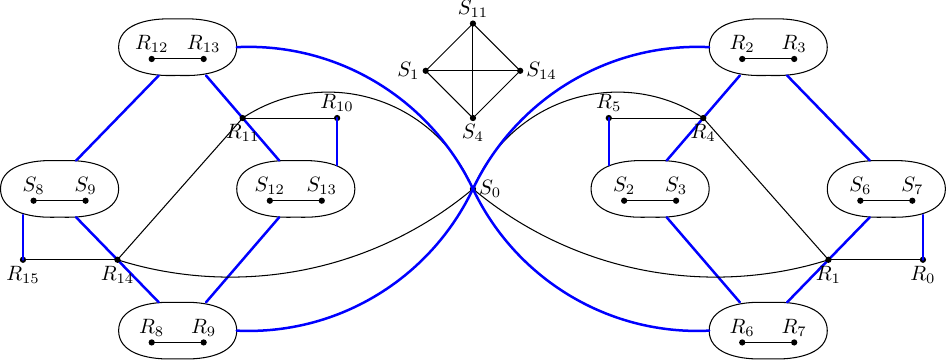}
\end{center}
\caption{The graph of vertically adjacent fooling rectangles.
Blue lines represent multiple edges.}\label{pi2}
\end{figure}

\begin{proof}
This lemma for $M$ is verified directly by looking at the matrix of the function. For $U$
the proof is as follows. 

Let us write out the second projections of fooling rectangles:
\newcommand{\pr}{\mathrm{pr}}
\begin{align*}
\pr_2 R_{ab\bar a d}&= \{\bar b\phi\mid \phi(\Lambda)=a,\phi(b)=\bar a\}\\
\pr_2 R_{aba\bar b}&=\{*\phi\mid \phi(\Lambda)=a, \phi(b)=a\}\\
\pr_2R_{abab}&= \{\bar b\phi\mid \phi(\Lambda)=a,\phi(b)=a,\phi(\bar b)=\bar a\}\\
\pr_2S_{ab\bar ad}&=\{b\phi\mid \phi(\Lambda)=a,\phi(b)=\bar a\}\\
\pr_2S_{aba\bar b}&= \{\rr\}\\
\pr_2S_{0}&=\{i\phi\mid\phi(\Lambda)=\phi(i)\}.
\end{align*}
To prove the lemma, we will make
several simple observations, immediately following from the definitions:
\begin{enumerate}
\item The graph splits into two connected components:
$\{S_1, S_4,S_{11},S_{14}\}$ (their second projection contains only $\rr$)
and all other vertices. 
The first component is a complete graph on vertices $\{S_1, S_4,S_{11},S_{14}\}$.
\item Rectangle $S_0$ is adjacent to all vertices $R_i$, except $R_0,R_5,R_{10},R_{15}$.

\item If we remove the rectangle $S_{0}$, then the second component again splits
into two connected components. One of them
consists of all rectangles of the form $R_{0bcd},S_{0bcd}$,
and the other one of all rectangles of the form $R_{1bcd},S_{1bcd}$. We will call these components
\emph{the second} and \emph{the third} ones. The second and third components are isomorphic, they are obtained from each
other by swapping zero and one. Therefore, it is sufficient to describe only the second component.

\item Second projections of rectangles
$R_{ab\bar a0},R_{ab\bar a1}$ coincide. Therefore
these rectangles are adjacent. Besides,  from them the edges lead to the same vertices.
The same applies to pairs of the form $S_{ab\bar a0},S_{ab\bar a1}$
In particular, the second component contains the edges $(R_0,R_1)$, $(R_2,R_3)$,
$(S_2,S_3)$, $(R_4,R_5)$, $(R_6,R_7)$,
$(S_6,S_7)$.

\item Finally, the second component contains
all the edges of the complete bipartite graph with the first part $R_0,R_1,R_2,R_3,S_2,S_3$
and the second part  
$R_4,R_5,R_6,R_7,S_6,S_7$ with the exception of the edges $(R_0,R_4)$, $(R_0,R_5)$, $(R_1,R_5)$
and all edges of a complete bipartite graph with the first part
$R_0,R_2,R_3,S_6,S_7$ and the second part $R_5,R_6,R_7,S_2,S_3$.

In the first part, we listed  all rectangles  $R_{abcd}$ from the second component
with $a=b=0$. In the second part, those with $a=0 $ and $b=1$.
Let us verify first
that $\phi$-projections of all the listed pairs 
intersect, that is, the corresponding conditions on $\phi$ are compatible.
Let $R(S)_{00cd}$ and $R(S)_{01c'd'}$ be rectangles 
from the second component. Then
requirements for $\phi(\Lambda)$ in $\phi$-the projections for both rectangles
coincide. And because  the  second bits in $00cd$ and $01c'd'$ are
different, the conditions on $\phi(b)$ are compatible.
Finally, if the first of these rectangles is $R_{0000}$ or the second rectangle is $R_{0101}$ (or both),
then there is an additional requirement $\phi(\bar b)=\bar a$.
These conditions may be incompatible 
for rectangles  $R(S)_{00cd}$ and $R(S)_{01c'd'}$. It is not difficult to verify that this happens only for
the pairs $(R_0,R_4)$, $(R_0,R_5)$, $(R_1,R_5)$. For this reason, these edges were excluded.

It remains to verify that the first symbols can also be the same.
For rectangles of the form $R_{aba\bar b}$, the first characters can be arbitrary.
For rectangles $R_0,R_2,R_3,S_6,S_7$ the first symbol is equal to 1,
and for rectangles $R_5,R_6,R_7,S_2,S_3$ the first symbol is equal to 0.
Because of this, edges between rectangles of the first and the second lists were excluded.
\end{enumerate}
\end{proof}

The graph in Fig.~\ref{pi2} consists of a complete graph on vertices $S_1,S_{11},S_4,S_{14}$
and two isomorphic graphs on 12 vertices together with the vertex $S_0$
 connecting them.
In the sequel, we will call these two parts of the 12 vertices \emph{the components} of the graph.
The smallest set of neighbors of a set of 13 vertices is attained, for example, if
we take all 12 vertices from the left component and the vertex $R_0$ from the right component,
or any 9 vertices from the left component and vertices $S_1,S_4,S_{11},S_{14}$.

Let us prove that this is indeed the minimal neighborhood of any set of 13 nodes.
To do this, we need the following auxiliary notion.
Let $A$ be an arbitrary set of vertices from one component of the graph.
Denote by $\Gamma(A)$ the set of all neighbors of $A$, and by $\Gamma'(A)$ the set of all neighbors, not counting the vertex $S_0$,
that is, the set of all neighbors inside that component.
By $\Gamma(n)$ we denote the minimum $|\Gamma(A)|$ for all $n$-element sets inside one component.
$\Gamma'(n)$ is defined in a similar way. The values of these functions are given in Table~\ref{t1},
they were found by hand. In fact, we will need the values of these functions only
for $n=1,4,6,7,9$.
\begin{table}
\begin{center}
\begin{tabular}{|c|c|c|}
\hline
$n$& $\Gamma'(n)$&$\Gamma(n)$\\
\hline
1& 4&4\\
\hline
2& 5&6\\
\hline
3& 6&6\\
\hline
4& 7&8\\
\hline
5& 7&8\\
\hline
6& 9&10\\
\hline
7& 10&11\\
\hline
8& 11&12\\
\hline
9& 12&13\\
\hline
10& 12&13\\
\hline
11& 12&13\\
\hline
12& 12&13\\
\hline
\end{tabular}
\end{center}
\caption{The table of values of functions $\Gamma$ and $\Gamma'$
expressing the expansion properties of components of the vertical adjacency graph.}\label{t1}
\end{table}

\begin{lemma}\label{l7}
(a) $\Gamma'(1)\ge 4$\\
(b) $\Gamma'(4)\ge 7$\\
(c) $\Gamma'(6)\ge 9$,\\
(d) $\Gamma'(7)\ge 10$,\\
(e) $\Gamma(9)\ge 13$.
\end{lemma}

Let us finish the proof of Lemma~\ref{l5} assuming Lemma~\ref{l7}.
Let a set $A$ consist of 13 vertices of the adjacency graph. We need to prove
that $A$ has at least 17 neighbors. Let us consider two cases.

\emph{Case 1:} the set $A$ contains the vertex $S_0$.
First, assume that the set $A$ contains at least one of the vertices $S_1,S_4,S_{11},S_{14}$.
Then there must be at least $13-1-4=8$ vertices in the left and right components together.
Therefore, one of the components must contain at least $8/2=4$ vertices.
Thus the number of neighbors of $A$ is not less than $5+\Gamma'(4)+6\ge 5+7+6 =18$
neighbors (vertices $S_0,S_1,S_4,S_{11},S_{14}$
plus all neighbors inside the component containing 4 vertices, plus the 
neighbors of vertex $S_0$ inside the other component).

Otherwise the set $A$ does not contain any of the vertices $S_1,S_4,S_{11},S_{14}$.
Then,  there should be at least $13-1=12$ vertices in the left and right components in total.
If there are 6 vertices in both components, then $A$ has at
least $1+2\Gamma'(6)\ge 1+2\cdot9=19$ neighbors.
Otherwise, one of the components has at least $12-5=7$ vertices, and the number of neighbors is not less
$1+\Gamma'(7)+6\ge 1+10+6 =17$ (the vertex $S_0$
plus all the neighbors inside the component containing 7 vertices, plus
the neighbors of $S_0$ inside the other component).

\emph{Case 2: } $S_0\notin A$.
Assume first that $A$ has at least one of the vertices $S_1,S_4,S_{11},S_{14}$.
In total, both components have at least $13-4=9$ vertices from the set $A$.
If all these vertices are in one component,
then the total number of $A$'s neighbors is at least
$\Gamma(9)+4\ge 13+4=17$ (neighbors of these 9 vertices plus vertices $S_1,S_4,S_{11},S_{14}$).
If each component has at least one vertex from $A$
and one of the components has at least 6 vertices from $A$, then the number of $A$'s
neighbors is at least $\Gamma(6)+\Gamma'(1)+4\ge\Gamma'(6)+\Gamma'(1)+4\ge 9+4+4 =17$.
Otherwise, there are 5 vertices in one component, and 4 in the other,
and the number of $A$'s neighbors is at least 
$\Gamma(5)+\Gamma'(4)+4\ge \Gamma'(4)+\Gamma'(4)+4\ge7+7+4=18$.

It remains to consider the case when $S_0,S_1,S_4,S_{11},S_{14}\notin A$.
Since there are only 12 vertices
in each component, each component must have at least one vertex from $A$.
If one of them has at least $9$ vertices from $A$, then we get at least
$\Gamma(9)+\Gamma'(1)\ge13+4=17$ neighbors. Otherwise, the division of vertices between components is 
$5+8$ or $6+7$. In both cases
the number of neighbors is not less than $\Gamma(5)+\Gamma'(7)\ge\Gamma'(4)+\Gamma'(7)\ge7+10=17$.
Lemma~\ref{l5} is proved (modulo Lemma~\ref{l7}).
\end{proof}

\begin{proof}[Proof of Lemma~\ref{l7}]
(a) This can be directly verified by hand.

(b) Assume that there is a set $A$ of 4 vertices of the right component that has less than 7 neighbors inside
  that component. We want to derive a contradiction.

Let us first assume that $A$ does not contain any vertices from the ovals.
Then $A=\{R_0,R_1,R_4,R_5\}$ and has $12\ge7$ neighbors inside the component.
So $A$ must contain a vertex from one of the ovals
and without loss of generality another vertex from the same oval
(indeed, adding this other vertex to $A$ does not increase the number of neighbors,
so it can replace any other vertex in $A$).
Identifying symmetric variants, we can consider only two cases:
$R_2,R_3\in A$ and $S_2,S_3\in A$.

\emph{First case:} $R_2,R_3\in A$. These two vertices together
already have 5 neighbors $R_2,R_3,R_4, S_6,S_7$.
It is easy to verify that adding
to $R_2,R_3$ any other vertex increases the number of neighbors by at least 2.

\emph{The second case:} $S_2,S_3\in A $. These two vertices together
already have neighbors $S_2,S_3,R_4, R_5,R_6,R_7$. It is easy to see that adding to $S_2,S_3$ any vertex, except $R_5$, increases the number of neighbors. Since we need
to add two vertices, we get a contradiction.

(c) To reduce the number of cases, we will consider  the complements. Let us state this technique in general.
We claim that the inequality $\Gamma'(m)\ge n$ implies the inequality
$\Gamma'(13-n)\ge 13-m$. Indeed, assume that $\Gamma'(13-n)< 13-m$.
That is, there is a set $A$ of cardinality $13-n$ having at most $12-m$
neighbors. Then consider the set $B$ consisting of non-neighbors $A$, $|B|\ge 12-(12-m)=m$.
Then all the neighbors of $B$ lie in the complement of $A$, so $|\Gamma'(B)|\le 12-(13-n)=n-1$.
By removing $|B|-m$ vertices from $B$, we get a set of $B'$ of cardinality exactly $m$
with $|\Gamma'(B')|\le n-1$.

It is worth to apply this technique when $13-n$ is closer to 6 than $m$,
because in this case the number of $(13-n)$-element subsets
of a component is less than the number of its $m$-element subsets. For example,
the statement (c) follows from statement (b), and we do not have to look over all $6$-element
subsets of a component.

(d) Using the same technique, it suffices to prove the inequality  $\Gamma'(3)\ge 6$.
Assume that there is a set $A$ of 3 vertices of a component that has less than 6 neighbors inside
that component.
We want to derive a contradiction.

First, assume that $A$ does not contain any vertices from the ovals.
Identifying symmetric variants, there are only two such sets:
$\{R_0,R_1,R_4\}$ and $\{R_0,R_1,R_5\}$. In the first case, all vertices of the components
are neighbors, and in the second case all vertices except $R_2,R_3$.

Now assume that $A$ contains a vertex from an oval. Without loss of generality, we can assume
that the other vertex of the oval also lies in $A$.

Identifying symmetrical cases, only two cases can be considered:
$A\ni R_2,R_3$ and $A\ni S_2,S_3$.
In the first case , the vertices
$R_2,R_3$ already have in total 5 neighbors $R_2,R_3,R_4, S_6,S_7$, so the third vertex
of $A$ should be in this list and should not have
neighbors out off the list. Be a  simple search we can see that there are no such vertices.
In the second case, the vertices $S_2,S_3$ already have 6 neighbors in total.

(e) Using our technique, we can deduce from item  (a)
that $\Gamma'(9)\ge 12$. In addition, the vertex $S_0$ is
a neighbor of any set of $12-6+1=7$ vertices inside one component,
since it has 6 neighbors inside one component. Therefore, for all $n$ starting from $n=7$, we have
$\Gamma'(n)=\Gamma(n)+1$.
\end{proof}

\end{document}